\documentclass[11pt]{article}

\usepackage[a4paper, left=1in, right=1in, top=1in, bottom=1in]{geometry}

\usepackage{fontspec}
\usepackage{amsmath,amssymb,amsthm,mathtools}
\usepackage{graphicx}
\usepackage{adjustbox}
\usepackage{tikz}
\usetikzlibrary{arrows.meta,positioning}
\usepackage{booktabs}
\usepackage{array}
\usepackage{mathpartir}
\usepackage{mathrsfs}
\usepackage{stmaryrd}
\usepackage{xcolor}
\usepackage{url}
\usepackage{algorithm}
\usepackage{algpseudocode}
\usepackage{placeins}
\usepackage{minted}
\usepackage{hyperref}
\usepackage{aliascnt}
\usepackage[capitalise,nameinlink,noabbrev]{cleveref}

\usepackage{newunicodechar}

\newfontfamily\coqsymbolfont{JuliaMono-Regular.ttf}[Scale=MatchLowercase]

\newcommand{\setcoqmonofont}{%
  \setmonofont{FreeMono.otf}[
    Scale=MatchLowercase,
    BoldFont=FreeMonoBold.otf,
    ItalicFont=FreeMonoOblique.otf,
    BoldItalicFont=FreeMonoBoldOblique.otf
  ]%
}

\newunicodechar{⦃}{{\coqsymbolfont\char"2983\relax}}
\newunicodechar{⦄}{{\coqsymbolfont\char"2984\relax}}

\hypersetup{hidelinks}
\makeatletter
\providecommand{\theHALG@line}{}
\renewcommand{\theHALG@line}{\thealgorithm.\arabic{ALG@line}}
\makeatother

\setminted{
  fontsize=\scriptsize,
  breaklines=true,
  breakanywhere=true,
  autogobble=true,
  xleftmargin=2em,
  numbersep=0.5em
}

\newcommand{\CameraReady}{}

\theoremstyle{definition}
\theoremstyle{definition}

\newtheorem{definition}{Definition}[section]

\newaliascnt{theorem}{definition}
\newtheorem{theorem}[theorem]{Theorem}
\aliascntresetthe{theorem}

\newaliascnt{lemma}{definition}
\newtheorem{lemma}[lemma]{Lemma}
\aliascntresetthe{lemma}

\newaliascnt{corollary}{definition}

\aliascntresetthe{corollary}

\newcommand{\cA}{\mathcal{A}}
\newcommand{\cB}{\mathcal{B}}

\newcommand{\cP}{\mathcal{P}}
\newcommand{\cQ}{\mathcal{Q}}
\newcommand{\cM}{\mathcal{M}}

\newcommand{\bbR}{\mathbb{R}}
\newcommand{\bbN}{\mathbb{N}}
\newcommand{\bbZ}{\mathbb{Z}}
\newcommand{\parens}[1]{\left(#1\right)}
\newcommand{\norm}[1]{\left\lVert#1\right\rVert}
\newcommand{\dKL}{D_{\mathrm{KL}}}
\newcommand{\KL}[2]{\dKL\!\parens{#1\,\|\,#2}}
\newcommand{\DG}[2]{D_{\bbZ,#1,#2}}
\newcommand{\distr}{\mathsf{Distr}}
\newcommand{\semantic}[1]{\llbracket #1\rrbracket}
\newcommand{\csemantic}[1]{\llbracket #1\rrbracket^{\bot}}
\newcommand{\completed}[1]{\mathsf{Comp}\parens{#1}}
\newcommand{\supp}{\mathsf{supp}}
\newcommand{\TVD}{\Delta}
\newcommand{\Pyth}{\mathsf{Pyth}}
\newcommand{\AEJ}{\mathsf{AE}}
\newcommand{\WinProb}{\mathsf{WinProb}}

\newcommand{\code}{\mathsf{code}}

\newcommand{\doubleplus}{\mathbin{+\!\!+}}

\title{Verified Pythagorean Composition for \\ Adaptive Cryptographic Games:\\
Noise Flooding in Homomorphic Encryption}

\author{
Yi Lee\thanks{Joint Center for Quantum Information and Computer Science, University of Maryland, USA. Email: \url{ylee1228@umd.edu}} 
\and
Alexandru Cojocaru\thanks{School of Informatics, The University of Edinburgh, United Kingdom. Email: \url{a.cojocaru@ed.ac.uk}}
\and
Junyi Liu\thanks{Joint Center for Quantum Information and Computer Science, University of Maryland, USA. Email: \url{junyiliu@umd.edu}}
\and
Xiaodi Wu\thanks{Joint Center for Quantum Information and Computer Science, University of Maryland, USA. Email: \url{xwu@cs.umd.edu}}
}

\date{}

\begin{document}

\maketitle

\begin{abstract}
Noise flooding is a standard defense against decryption attacks on approximate
homomorphic encryption, but its security proof is unusually sensitive to
composition.
Replacing each of \(q\) adaptive decryption answers with a
statistically close simulation and applying an ordinary hybrid argument loses
linearly in \(q\).
The cryptographic proof instead accumulates conditional
Kullback--Leibler (KL) costs and converts to statistical distance once, giving
the parameter-critical square-root loss.

We machine-check this argument using Rocq and SSProve.
Given any fully homomorphic encryption scheme that is approximately correct and IND-CPA secure,
we formalize a reduction for every \(q\)-query IND-CPAD adversary and prove
\[
 \Pr[\mathsf{IND\text{-}CPAD}_{\mathsf{NF}}^{\mathcal A}=1]
 \leq \beta_{\mathsf{CPA}}(\mathcal B_{\mathcal A,q})
       + \frac{\sqrt{qn}}{2\gamma}.
\]
where \(n\) is the plaintext dimension and \(\gamma\) is the flooding-width
multiplier.
Our proof constructs a new relational program logic
over SSProve semantics.
Its Pythagorean judgment composes
conditional KL budgets without converting them to statistical distance, and a
verified trace compiler lifts a local oracle rule to arbitrary adaptive
programs with a single final conversion.
\end{abstract}

\section{Introduction}
\label{sec:introduction}

\begin{figure*}[t]
\centering
\begin{adjustbox}{max width=\linewidth}
\begin{tikzpicture}[
  x=1cm,
  y=1cm,
  >={Latex[length=2mm]},
  contribution/.style={->,semithick,draw=black!65},
  module/.style={
    draw=black!65,
    rounded corners=2pt,
    align=center,
    text width=3.0cm,
    minimum height=1.35cm,
    inner sep=3pt,
    font=\footnotesize
  },
  interface/.style={module,fill=orange!10},
  application/.style={module,fill=green!10},
  reusable/.style={module,text width=3.35cm,fill=blue!8},
  result/.style={module,fill=violet!12,draw=black!80,very thick},
  lane/.style={font=\sffamily\footnotesize\bfseries,anchor=south west}
]

\node[interface] (interfaces) at (1.55,4.35) {%
  \textbf{Assumptions\\on the scheme}\\[0.2ex]
  IND-CPA security,\\
  charts, and\\
  approximate\\
  correctness};

\node[lane] at ([yshift=0.3cm]interfaces.north west)
  {NOISE-FLOODING INSTANTIATION};

\node[application] (games) at (5.15,4.35) {%
  \textbf{Noise-flooded\\construction}\\[0.2ex]
  IND-CPA and\\
  IND-CPAD games\\
  and reduction};

\node[application] (local) at (8.75,4.35) {%
  \textbf{Analysis of one\\decryption call}\\[0.2ex]
  Gaussian KL bound\\
  and reachable-state\\
  invariant};

\node[application] (adaptive) at (12.35,4.35) {%
  \textbf{Adaptive replacement}\\[0.2ex]
  Replace up to \(q\)\\
  adaptively chosen calls};

\node[result] (final) at (15.95,4.35) {%
  \textbf{Checked\\security theorem}\\[0.2ex]
  \(\beta_{\mathsf{CPA}}(\mathcal B)+
    \sqrt{qn}/(2\gamma)\)};

\draw[contribution] (interfaces) -- (games);
\draw[contribution] (games) -- (local);
\draw[contribution] (local) -- (adaptive);
\draw[contribution] (adaptive) -- (final);

\node[reusable] (kl) at (1.85,1.35) {%
  \textbf{Discrete probability}\\[0.2ex]
  KL divergence, Pinsker,\\
  and conditional-\\
  coordinate bounds};

\node[lane] at ([yshift=0.3cm]kl.north west)
  {REUSABLE VERIFIED MACHINERY};

\node[reusable] (gaussian) at (6.15,1.35) {%
  \textbf{Discrete-Gaussian\\analysis}\\[0.2ex]
  Normalization,\\
  moments, and exact KL};

\node[reusable] (logic) at (10.45,1.35) {%
  \textbf{Pythagorean\\relational logic}\\[0.2ex]
  Accumulate KL\\
  costs and convert to\\
  additive error once};

\node[reusable] (compiler) at (14.75,1.35) {%
  \textbf{From one call to\\adaptive programs}\\[0.2ex]
  Selected-call compiler\\
  and semantic correctness};

\draw[contribution] (kl) -- (gaussian);
\draw[contribution]
  (kl.south) -- ++(0,-0.35)
  -| ([xshift=-0.4cm]logic.west) -- (logic.west);

\draw[contribution] (gaussian.north) -- (local.south);
\draw[contribution] (logic.north) -- ([xshift=-0.45cm]adaptive.south);
\draw[contribution] (compiler.north) -- ([xshift=0.45cm]adaptive.south);

\end{tikzpicture}
\end{adjustbox}
\caption{Overview of the checked development.  The upper lane specializes the
reusable results in the lower lane to noise flooding and culminates in the
exported security theorem.  Arrows show how the mathematical and program-level
results combine: orange marks assumptions, green the noise-flooding
instantiation, blue reusable verified machinery, and purple the final theorem.}
\label{fig:theory-map}
\end{figure*}
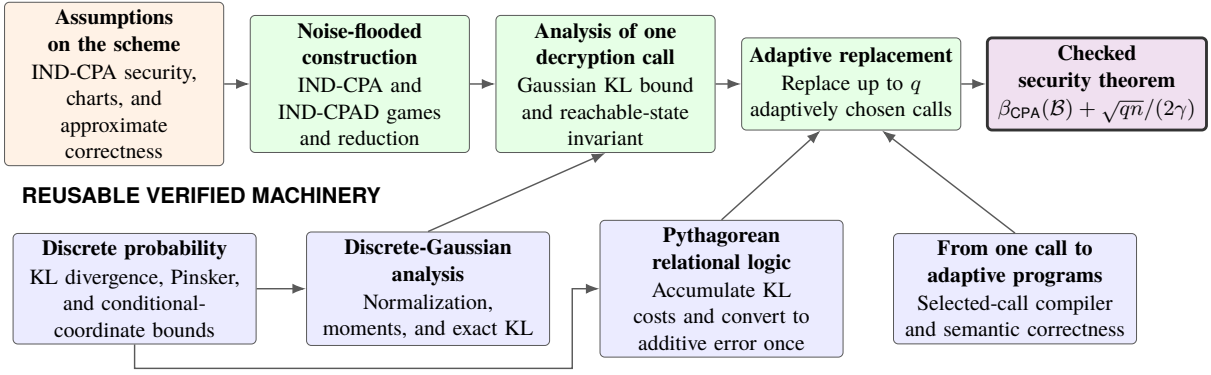

Ordinary encryption protects data while it is stored or transmitted, but
computation typically requires exposing the plaintext.
Fully homomorphic encryption (FHE) removes this boundary: a client can encrypt
its data and send the ciphertexts to an untrusted server, which computes on them
without the secret key and returns an encryption of the corresponding result.
First envisioned in 1978 and first constructed by Gentry in 2009, FHE was a
long-standing goal of cryptography~\cite{rad78,gentry09}; the LWE-based
constructions that launched the modern line of efficient schemes were
recognized by the 2022 G{\"o}del Prize~\cite{godel2022}.
Modern families target complementary workloads: BGV and BFV provide packed
modular arithmetic, TFHE targets Boolean computation and fast bootstrapping, and
CKKS provides packed approximate numerical
arithmetic~\cite{godel2022,fv12,tfhe,ckks}.
Within approximate FHE, CKKS is the dominant design and is implemented in
nearly every major library offering approximate arithmetic~\cite{heprofiler}.

CKKS is attractive for data-analysis and machine-learning workloads because
many numerical computations can tolerate small errors.
It represents values at finite precision and allows error to grow in a
controlled fashion during encrypted computation, so decryption returns a value
close to, rather than exactly equal to, the intended result.
In some applications, however, the server can obtain decrypted results through
repeated interaction with the client.
An attacker who knows the inputs and the computation can then predict its
mathematical result and compare that prediction with the approximate value
returned after decryption.
Li and Micciancio showed that these small discrepancies can depend on the secret
key strongly enough to recover it~\cite{lm}.

Li and Micciancio captured the missing security guarantee with the
\emph{IND-CPAD} security game.
It grants only restricted decryption access: the attacker may query ciphertexts
produced through legitimate encryption and evaluation, and only when it already
knows the mathematical result.
For a perfectly correct exact scheme, such a decryption reveals nothing beyond
that known result.
CKKS instead reveals an approximate result, including the residual error
exploited by the attack.
Li, Micciancio, Schultz-Wu, and Sorrell (LMSS) subsequently showed that
approximate FHE can be secured in this model by adding fresh, sufficiently wide
Gaussian noise to each released result~\cite{lmss}.
This defense, called \emph{noise flooding}, masks the original
secret-key-dependent error but also consumes precision, so its width must balance
security against the usefulness of the decrypted
result~\cite{costache2023precision,bergamaschi2025revisiting}.

The security--precision tradeoff depends critically on how the loss composes
across repeated decryptions.  Security proofs for adaptive cryptographic
systems often reduce a global attack to a local distributional replacement:
one oracle answer in the real game is replaced by a nearby simulated answer.
For noise flooding, the real answer adds fresh noise around the scheme's
approximate decryption, whereas the simulated answer adds noise with the same
width around the known exact result.  Approximate correctness bounds the
displacement between these centers.  For every reachable history in which a
permitted query is answered, an equal-variance discrete-Gaussian calculation
then bounds the conditional KL divergence between the two answer distributions
by
\[
  \varepsilon_{\mathrm{nf}} = \frac{n}{2\gamma^2},
\]
where \(n\) is the noise-coordinate dimension and \(\gamma\) controls the
flooding width~\cite{lmss}.

An attacker may make up to \(q\) queries, choosing each one after seeing the
previous answers, so the analysis must compare entire adaptive transcripts.
Applying Pinsker's inequality separately to each replacement and then adding
the resulting distances with the triangle inequality gives the loss
\[
  q\sqrt{\frac{\varepsilon_{\mathrm{nf}}}{2}}
  = \frac{q\sqrt{n}}{2\gamma}.
\]
For the sharper analysis, let \(\mathcal P\) and \(\mathcal Q\) be the
distributions of the complete real and simulated transcripts.  Because the KL
bound holds conditionally after every reachable history, the KL chain rule
adds the costs of the adaptively chosen answers.  Pinsker's inequality is then
applied once, to the complete transcripts, giving~\cite{lmss,mw17}
\begin{align*}
  \KL{\mathcal P}{\mathcal Q}
    &\leq q\varepsilon_{\mathrm{nf}},\\
  \TVD(\mathcal P,\mathcal Q)
    &\leq \sqrt{\frac{q\varepsilon_{\mathrm{nf}}}{2}}
     = \frac{\sqrt{qn}}{2\gamma}.
\end{align*}
Data processing transfers this transcript bound to the games' output
distributions and hence to their winning probabilities.
Thus, for a fixed security bound, the first analysis requires flooding width
growing linearly in \(q\), whereas the second requires growth only as
\(\sqrt q\).  Because the flooding noise is added directly to the decrypted
result, the sharper composition bound directly preserves numerical precision.

This chain-rule composition principle predates approximate FHE.
Micciancio and Walter studied cryptographic algorithms whose analyses assume
access to an ideal sampler, even though a concrete implementation may only
realize an approximation---for example, because it computes sampling
probabilities using finite-precision floating-point arithmetic~\cite{mw17}.
They called the resulting root-sum-of-squares behavior \emph{Pythagorean
probability preservation}.
Thus, Pythagorean preservation is not specific to FHE: it is the general
composition step connecting a one-call approximation guarantee to an adaptive
computation.

Although this composition step is concise on paper, it combines probability
calculations, adaptive interactions with an attacker, and transformations
between security games.
Small omissions can be difficult to detect on paper, yet may change either the
security claim or its concrete parameters.
We therefore machine-check the argument.

To machine-check a proof, we translate its prose and equations into a precise
language understood by a specialized software system called a \emph{proof
assistant}.
Our development uses the Rocq proof assistant together with SSProve, a Rocq
library for cryptographic games.
The proof remains human-written, but Rocq checks that each step follows from the
stated assumptions and earlier results.
When we say below that we verify a theorem, we mean that we construct such a
machine-checked proof.

One influential approach to machine-checking game-based security proofs
represents each game as a probabilistic program.
The challenger and attacker sample randomness, maintain state, and call
procedures, while a proof step relates one such program to another.
CertiCrypt established this code-based approach inside Coq; EasyCrypt
subsequently pursued a dedicated language and greater
automation~\cite{certicrypt,easycrypt}.
SSProve takes a modular approach inside Rocq: it represents games as packages
that can be linked and replaced, and provides tools for comparing the procedures
inside those packages~\cite{ssprove}.
Its extensibility is central to our work.
We can add a new way to compare two probabilistic programs and use Rocq to prove
that it gives only valid conclusions, while reusing SSProve's existing
representation of cryptographic games.
This architecture fits the two levels of the LMSS proof: a reduction between
complete security games and a probabilistic comparison between two
implementations of a single decryption query.

Our mechanization follows the two levels just described.
At the cryptographic level, it checks the LMSS noise-flooding reduction under
the assumptions stated below, reducing the IND-CPAD winning probability to an
IND-CPA bound plus an explicit square-root loss.
At the program level, it proves a reusable local-to-adaptive replacement
theorem: given a conditional KL bound and a preserved invariant for one oracle
call, the theorem compares complete programs containing up to \(q\) adaptively
chosen calls.
A verified selected-call compiler makes arbitrary SSProve adversary code
compatible with this theorem by exposing the selected calls while preserving
the surrounding adaptive computation.  Combining the compiler with our
Pythagorean judgment gives the quantitative replacement theorem.
Noise flooding is the main checked application, but the replacement theorem
itself does not mention FHE.

\subsection{Machine-checked main theorem}
\label{sec:scope}

Let \(S\) be an approximate-FHE scheme whose message space has local charts in
\(\bbZ^n\), let \(\mathsf{NF}_{\gamma}[S]\) be its noise-flooded transform, and
let \(\cA\) make at most \(q\) non-failing decryption queries in the IND-CPAD
game.
Under the interfaces and assumptions specified in
\Cref{sec:construction}, the formal development constructs an IND-CPA
adversary \(\cB_{\cA,q}\) and proves
\begin{equation}
\label{eq:headline}
 \WinProb[\mathsf{IND\text{-}CPAD}_{\mathsf{NF}_{\gamma}[S]}(\cA)]
 \leq
 \beta_{\mathsf{CPA}}(\cB_{\cA,q})+
 \frac{\sqrt{qn}}{2\gamma}.
\end{equation}
The expression \(\beta_{\mathsf{CPA}}\) is the supplied winning-probability
bound for the underlying scheme.
The formal theorem is given in the artifact as
\texttt{NoiseFloodingSecure.is\_secure}.

The result covers adaptive oracle programs rather than adversaries prewritten as \(q\) separate query phases.
Its proof isolates one lossy transition: real flooded decryption is replaced by flooding around the plaintext recorded in the challenge table.
At the cryptographic level this is the single \(G_0\to G_1\) hop of LMSS.
The formal proof exposes this transition by opening and relinking packages,
compiling calls, replacing unreachable residual decryptions, and identifying
the reduction game; each of these representation changes is proved exact.

The checked theorem deliberately isolates the good-execution step of the LMSS
argument.
It assumes deterministic decryption, approximate correctness with certainty
(that is, no bad key pairs or encryptions exist), and a message space with local
\(\bbZ^n\) charts.
A concrete CKKS instantiation would additionally require a standard up-to-bad game hop
for correctness failures and proofs that CKKS satisfies IND-CPA and the stated
interface assumptions.
We have not yet formalized those obligations, and therefore do not claim a verified
CKKS deployment or concrete parameter set.

As in SSProve generally, probabilistic polynomial time is not formalized in the
Rocq semantics and remains a conventional external
audit~\cite{ssprove}.
Here the standard reduction runs \(\cA\) once and adds oracle forwarding, table
bookkeeping, and Gaussian sampling; no unusual complexity issue is hidden in
the reduction.

\subsection{Why a program logic?}

The noise-flooding proof starts with a KL-divergence bound between two
implementations of one decryption call and must end with an additive bound
between two complete games containing an adaptive adversary.
SSProve's original architecture addresses this kind of local-to-global gap by
connecting two proof styles: high-level equational reasoning about
cryptographic packages and low-level relational reasoning about their
procedures~\cite{ssprove}.
Noise flooding needs a new quantitative connection of the same kind.

Concretely, an SSProve procedure is represented inside Rocq as an abstract
syntax tree (AST).
Its nodes record instructions to sample randomness, update state, and call
procedures, but the tree is not itself a probability distribution.
The probability lemmas used in the security proof instead apply to the
distribution over outputs and states produced when the procedure runs.
Thus a proof must connect the syntax of the game to the distributions about
which its probabilistic claims are made.

The conventional solution in program verification is a \emph{program logic}: a
collection of reusable proof rules, in a tradition dating back to Hoare
logic~\cite{hoare69}.
The logic records properties of program fragments as \emph{judgments}, and its
rules show how these properties behave when fragments sample randomness, are
run in sequence, or call procedures.
Proving that these rules give correct conclusions about the program's output distributions
establishes the logic's \emph{soundness}.
A relational program logic reasons about two programs at a time, making it
suited to transformations between cryptographic games.
Probabilistic relational program logics already underlie both EasyCrypt and
SSProve~\cite{easycrypt,ssprove}.

An ordinary approximate-equivalence judgment does not preserve the required parameter dependence.
Its sequence rule adds scalar errors, so applying Pinsker at every oracle call turns a one-call cost \(\sqrt{\epsilon/2}\) into \(q\sqrt{\epsilon/2}\).
Instead, we define a \emph{Pythagorean judgment}, which records a tuple of conditional KL costs.
Sequencing concatenates tuples, and a checked Micciancio--Walter rule converts the accumulated tuple to additive error only once.
After \(q\) calls of KL cost \(\epsilon\), this gives \(\sqrt{q\epsilon/2}\).

The remaining obstacle is adaptivity: a concrete SSProve adversary is one stateful oracle program, not a syntactic list of \(q\) calls.
We therefore verify a compiler that exposes selected oracle calls while
preserving the surrounding control flow.  Its same-package correctness theorem
shows that the transformation preserves program behavior when selected and
residual calls use the same implementation (\Cref{thm:compile-same}).
Together, the judgment and compiler form a reusable interface: prove one
invariant-preserving oracle rule, then obtain a whole-program additive bound
for arbitrary well-typed adaptive code.
Noise flooding is the first substantial application of this interface;
\Cref{thm:compile-replace} states the general quantitative theorem.

The program logic is broader than the particular LMSS noise-flooding reduction
formalized here.
One direction for reuse returns to the original motivation of Micciancio and
Walter: comparing an ideal sampler with a finite-precision implementation inside
an adaptive cryptographic computation.
A second remains within FHE and noise flooding, but moves to threshold
protocols, where the secret key is shared and parties release partial
decryptions across repeated, possibly adaptive interactions.
Neither application is mechanized here; each would require a one-call
conditional KL bound and invariant suited to its setting.

\subsection{Contributions}

We separate the cryptographic result from the machinery used to establish it.

\textbf{A verified noise-flooding reduction.} We define the IND-CPA and IND-CPAD games, the noise-flooded construction, and a concrete reduction as SSProve packages.
Rocq checks the single statistical game hop, its exact representation-level
refinements, and the loss in \Cref{eq:headline}.
The construction, game reduction, and adaptive statistical step are presented
in \Cref{sec:security-model,sec:verified-reduction,sec:closing-hop}; full game
definitions and reduction pseudocode appear in \Cref{app:security-games}.

\textbf{A reusable Pythagorean program logic.}
We define \emph{additive-error} and \emph{Pythagorean} judgments, and prove
several inference rules in Rocq.
The logic combines unary invariant reasoning,
direct coupling and additive-error rules, KL sampling rules, tuple-concatenating
sequencing, and a final bridge to statistical distance.
We discuss our program logic in \Cref{sec:pyth-toolkit}.

\textbf{A verified selected-call compiler and local-to-adaptive theorem.}
We construct a trace-based compiler that exposes the first \(q\) calls to a
chosen operation inside an arbitrary concrete SSProve program.  We prove the
transformation semantically exact when selected and residual calls use the
same package.  Combining the compiler with one
invariant-preserving Pythagorean oracle judgment yields a whole-program
additive-error judgment with loss \(\sqrt{q\norm{s}_1/2}\).  The compiler,
exactness theorem, and quantitative replacement theorem are independent of
FHE and are presented in \Cref{sec:trace-compiler}.

\textbf{Probability analysis.} We verify KL divergence and its finiteness obligations, Pinsker's inequality, a conditional-coordinate preservation theorem, and the equal-variance KL formula for integer discrete Gaussians.
These results provide the underlying mathematical foundations used in the logic.
Their roles in the one-call Gaussian bound and the final Pythagorean conversion
appear in \Cref{sec:simulation-target,sec:pyth-toolkit}, respectively.

Across these contributions, the Rocq formalization is organized so that games,
assumptions, loss, compiler, and final theorem can be reviewed separately.
The overview in \Cref{fig:theory-map} makes the boundary between the reusable
probability and program-logic results and their noise-flooding application
explicit.  We report the trusted base, distinguish inherited library
assumptions from local proof obligations, and provide a verified replacement
for the one unverified real-analysis lemma inherited through the current
SSProve stack; \Cref{sec:artifact} gives this mechanization audit.

\ifdefined\CameraReady
Our Rocq formalization is available at
\url{https://github.com/ethanlee515/Mending}.
\fi

\section{Construction, Security Setting, and Result}
\label{sec:security-model}

This section defines the construction and states the result proved later,
explaining formalization-specific conventions as they arise.
All distributions are discrete.
Games return a Boolean, and
\(\WinProb[G]=\Pr[G=\mathsf{true}]\).
We write
\(\TVD(P,Q)=\tfrac12\sum_x|P(x)-Q(x)|\) for statistical distance and
\(\KL{P}{Q}\) for KL divergence when the defining series is summable and
\(P\) is absolutely continuous with respect to \(Q\).

\subsection{Approximate FHE and its flooded transform}
\label{sec:construction}

CKKS introduced homomorphic arithmetic with approximate plaintext
semantics~\cite{ckks}.
LMSS subsequently gave a generic approximate-correctness formulation in which
ciphertexts carry public error estimates~\cite{lmss}.
Our formal scheme interface is a typed, executable presentation of that
setting:
\[
 S=(\mathsf{KeyGen},\mathsf{Enc},\mathsf{Eval}_1,
       \mathsf{Eval}_2,\mathsf{Dec}).
\]
These algorithms are distribution-valued, and key generation is lossless: it
returns with probability one.
For bookkeeping, the formal model represents a ciphertext as either
\(\mathsf{None}\), for an unsupported evaluation, or
\(c=\mathsf{Some}(\bar c,e)\), pairing a ciphertext with its public error bound
\(e\in\bbN\).
This option wrapper belongs to the formal API, not to the cryptographic
construction.
Unary and binary evaluation gates suffice to describe arbitrary circuits while
keeping the package interface small.

Although the interface permits randomized decryption, the theorem specializes
it to a deterministic value \(\mathsf{dec}_0(sk,c)\) and assumes
probability-one approximate correctness:
for every encryption or evaluation output
\(c=\mathsf{Some}(\bar c,e)\) having nonzero probability and representing
\(m\), we have
\begin{equation}
\label{eq:approx-correctness}
 d(\mathsf{dec}_0(sk,c),m)\leq e.
\end{equation}
The artifact calls this \emph{support-level} correctness; in cryptographic
terms, the bound holds with probability one.
Nonzero correctness error would require the additional up-to-bad hop discussed
in \Cref{sec:scope}.

LMSS's noise-flooding defense postprocesses approximate decryption with fresh
Gaussian noise~\cite{lmss}.
On a concrete plaintext ring, adding such noise to a plaintext is an ordinary
algebraic operation.
For an abstract message type \(M\), however, a product discrete Gaussian is
supported on \(\bbZ^n\), not on \(M\); it cannot be used as message-valued
noise without first supplying integer coordinates.

Our formalization therefore packages the metric and the coordinates needed
for noise flooding as one interface.
The message space carries a distance \(d:M\times M\to\bbN\), and for each
center \(a\in M\) it has maps
\[
 I_a:M\rightarrow\bbZ^n,
 \qquad J_a:\bbZ^n\rightarrow M
\]
satisfying
\begin{equation}
\label{eq:charts}
 I_a(a)=0,
 \quad d(a,b)=\norm{I_a(b)}_\infty,
 \quad J_b(v)=J_a(v+I_a(b)).
\end{equation}
Here \(I_a(b)\) is the integer displacement from \(a\) to \(b\), while
\(J_a(v)\) interprets the integer displacement \(v\) as a message based at
\(a\).
Thus \(J_a\) transports a distribution on \(\bbZ^n\) to a distribution on
messages.
The three laws say that the coordinates are centered, that their
\(\ell_\infty\) norm agrees with message distance, and that changing the base
point translates the integer coordinates.
In particular, they let the proof compare flooding around two messages as
centered and shifted Gaussians in one common copy of \(\bbZ^n\).

For \(\gamma>0\), noise flooding around \(x\in M\) is consequently the
following single construction:
\begin{equation}
\label{eq:flood}
\begin{aligned}
 \mathsf{Flood}_e(x)&:\\
 &\sigma_e=\max\{1,e\}\gamma,\\
 &\eta\leftarrow D_{\bbZ^n,0,\sigma_e^2},\\
 &\mathsf{return}\;J_x(\eta).
\end{aligned}
\end{equation}
Here \(D_{\bbZ^n,0,\sigma_e^2}\) is the product of \(n\) independent
integer discrete Gaussians.
The sampling step lives in \(\bbZ^n\), and \(J_x\) makes its result a message;
an arbitrary metric space without such maps would not support
\Cref{eq:flood}.
This proof-specific interface is easier to audit than committing the generic
theorem to a concrete CKKS polynomial quotient ring.
The noise-flooded transform \(\mathsf{NF}_{\gamma}[S]\) keeps key generation,
encryption, and evaluation unchanged and decrypts a valid tagged ciphertext by
\[
 a\leftarrow\mathsf{Dec}(sk,c);
 \quad \mathsf{Flood}_e(a).
\]
On input \(\mathsf{None}\), decryption fails; SSProve represents this by a
subdistribution of total weight zero.

We use the standard multi-query left-or-right IND-CPA experiment and the
IND-CPAD experiment of Li and Micciancio~\cite{lm}; their full definitions
appear in \Cref{alg:ind-cpa-game,alg:ind-cpad-game,alg:ind-cpad-oracles}.
For the discussion below, the essential difference is that IND-CPAD adds
evaluation and restricted decryption: a decryption query is answered only
when the two plaintexts recorded for that ciphertext agree.
We write the assumed IND-CPA bound as
\[
 \WinProb[\mathsf{IND\text{-}CPA}_{S}(\cB)]
 \leq \beta_{\mathsf{CPA}}(\cB).
\]
Adversaries are well-typed SSProve oracle programs.  Running time is not part
of the semantics, so preservation of probabilistic polynomial time by the
explicit reduction remains a separate inspection~\cite{ssprove}.

\subsection{The one-query decryption replacement}
\label{sec:simulation-target}

The reduction compares two ways to answer a decryption query when the two
recorded plaintexts agree.
Fix a row that can occur with nonzero probability,
\[
 T_i=(m,m,c),
 \qquad c=\mathsf{Some}(\bar c,e).
\]
Under deterministic decryption, the real and simulated answers are
\begin{align}
\label{eq:real-sim-decrypt}
 \mathsf{ODec}_{\mathsf{real}}(i):\quad&
   y\leftarrow\mathsf{Flood}_e(\mathsf{dec}_0(sk,c));
   \ \mathsf{return}\ \mathsf{Some}(y),\notag\\
 \mathsf{ODec}_{\mathsf{sim}}(i):\quad&
   y\leftarrow\mathsf{Flood}_e(m);
   \ \mathsf{return}\ \mathsf{Some}(y).
\end{align}
Both versions perform the same bounds checks and counter updates, refuse the
same queries, and abort in the same circumstances.
The simulated answer needs neither \(sk\) nor \(b\): the common plaintext
\(m\) is already recorded in the table.

The probability-one correctness assumption in \Cref{eq:approx-correctness}
gives
\(d(\mathsf{dec}_0(sk,c),m)\leq e\).
By the change-of-basepoint law in \Cref{eq:charts}, the simulated answer can
be expressed in the real answer's chart by shifting the integer Gaussian by
\(I_{\mathsf{dec}_0(sk,c)}(m)\).
The distance law bounds the \(\ell_\infty\) norm of this shift by \(e\), so
each of its \(n\) coordinates has magnitude at most \(e\).
The equal-variance discrete-Gaussian identity
\begin{equation}
\label{eq:dg-kl}
 \KL{\DG{\mu}{\sigma^2}}{\DG{\nu}{\sigma^2}}
   =\frac{(\nu-\mu)^2}{2\sigma^2}.
\end{equation}
After fixing the adversary's prior view, \(m\), \(c\), and \(e\) are fixed.
The identity bounds the KL divergence \(\epsilon_i\) between the two answer
distributions for this call by
\begin{equation}
\label{eq:epsilon-nf}
 \epsilon_i
 \leq
 \frac{ne^2}{2\max\{1,e\}^2\gamma^2}
 \leq
 \epsilon_{\mathsf{nf}},
 \qquad
\epsilon_{\mathsf{nf}}=\frac{n}{2\gamma^2}.
\end{equation}

This is a one-call statement conditioned on an arbitrary prior interaction
history.
The central problem is to preserve it across the adversary's \(q\) adaptive
queries without converting to statistical distance after each call.

The underlying noise-flooding reduction is due to LMSS~\cite{lmss}.
Our contribution is to express its assumptions through the interfaces in
\Cref{sec:construction} and check the resulting adaptive reduction in Rocq.

\begin{theorem}[Checked noise-flooding reduction]
\label{thm:main}
Let \(S\), its noise-coordinate interface, probability-one correctness proof,
IND-CPA security interface, and \(\gamma>0\) satisfy
the requirements of \Cref{sec:construction}.
For every IND-CPAD adversary \(\cA\) represented as a well-typed SSProve oracle
program, and every \(q\in\bbN\), the formalization constructs an IND-CPA adversary
\(\cB_{\cA,q}\) with the required oracle interface and proves
\begin{align*}
 &\WinProb[
   \mathsf{IND\text{-}CPAD}_{\mathsf{NF}_{\gamma}[S],q}(\cA)]\\
 &\quad\leq \beta_{\mathsf{CPA}}(\cB_{\cA,q})
       +\sqrt{q\epsilon_{\mathsf{nf}}/2}\\
 &\quad\leq \beta_{\mathsf{CPA}}(\cB_{\cA,q})
       +\frac{\sqrt{qn}}{2\gamma}.
\end{align*}
\end{theorem}

The first inequality keeps the per-query bound
\(\epsilon_{\mathsf{nf}}\) visible; substituting
\(\epsilon_{\mathsf{nf}}=n/(2\gamma^2)\) from \Cref{eq:epsilon-nf} gives the
second.
Thus IND-CPA security is an assumption on \(S\), while IND-CPAD security of its
noise-flooded version is the conclusion.

\FloatBarrier

\section{The Reduction and Its Missing Link}
\label{sec:verified-reduction}

\Cref{sec:simulation-target} isolates the local change behind the LMSS
reduction: replace \(\mathsf{ODec}_{\mathsf{real}}\) by the
plaintext-centered \(\mathsf{ODec}_{\mathsf{sim}}\).
The latter can be implemented without the secret key and hence inside an
IND-CPA reduction.
For one call, \Cref{eq:epsilon-nf} bounds the KL cost.
We now lift that local comparison to the complete adaptive experiment.

\subsection{Three games}

Fix an IND-CPAD adversary \(\cA\) making at most \(q\) decryption queries.
Consider the following games:
\begin{align*}
 G_0 &=
   \mathsf{IND\text{-}CPAD}_{\mathsf{NF}_{\gamma}[S],q}(\cA),\\
 G_1 &=
   \mathsf{IND\text{-}CPAD}^{\mathsf{sim}}_
     {\mathsf{NF}_{\gamma}[S],q}(\cA),\\
 G_2 &=
   \mathsf{IND\text{-}CPA}_{S}(\cB_{\cA,q}).
\end{align*}
Game \(G_0\) is the real experiment.
Game \(G_1\) replaces the real decryption oracle by
\(\mathsf{ODec}_{\mathsf{sim}}\) from \Cref{eq:real-sim-decrypt}.
This simulated-decryption game is the intermediate game used by
LMSS~\cite{lmss}.
Game \(G_2\) runs the concrete reduction from
\Cref{alg:ind-cpa-reduction}.

The second transition is exact.
Once decryption answers are centered at the common recorded plaintext, they no
longer require the secret key or the hidden challenge bit.
The reduction can therefore forward encryption queries to its IND-CPA oracle,
perform evaluation and table bookkeeping locally, and simulate every permitted
decryption query with the same distribution as \(G_1\).
Consequently,
\[
 \WinProb[G_1]=\WinProb[G_2].
\]

The entire cryptographic proof is thus the chain
\begin{equation}
\label{eq:conceptual-game-chain}
 G_0
 \xrightarrow{\sqrt{q\epsilon_{\mathsf{nf}}/2}}
 G_1
 \xrightarrow{0}
 G_2.
\end{equation}
The first arrow is the only paid hop.
Together with the assumed IND-CPA bound for \(G_2\), it immediately implies
\Cref{thm:main}.

\subsection{The tempting one-line proof}
\label{sec:paid-hop}

\Cref{eq:epsilon-nf} bounds the conditional KL cost of every reachable
decryption query whose two recorded plaintexts agree by
\(\epsilon_{\mathsf{nf}}\).
The two games agree when those plaintexts differ and when an assertion fails,
so those branches spend no budget.

It is tempting to finish in one sentence:
apply the KL chain rule over the \(q\) decryption interactions, obtain total KL cost
\(q\epsilon_{\mathsf{nf}}\), and apply Pinsker once.
This gives
\begin{equation}
\label{eq:central-game-hop}
 \TVD(G_0,G_1)
 \leq \sqrt{q\epsilon_{\mathsf{nf}}/2}.
\end{equation}
This is the KL-composition argument used by LMSS~\cite{lmss}.
A machine-checked proof must identify the relevant conditional distributions
and establish the one-call bound after every adversarially chosen history.

\subsection{Where adaptivity enters}
\label{sec:adaptive-invariant}

The adversary's decryption calls do not form a fixed product experiment.
After seeing an answer, \(\cA\) may change its state, choose a different oracle,
or decide whether and where to make its next decryption query.
In SSProve, \(\cA\) is one stateful program with calls embedded in its control
flow, not a list of \(q\) pre-existing query phases.
Thus ``the first \(q\) decryption calls'' is a semantic notion, and the
conditional KL bound must hold after every reachable history.

The proof also depends on a state invariant \(\Phi\) relating the real and
simulated executions.
It records initialized, good keys; agreement of their public data and hidden
challenge bits; validity of the challenge table under encryption and
evaluation; and agreement of the decryption counters.
In particular, for every reachable row
\((m_0,m_1,c)\) with \(c=\mathsf{Some}(\bar c,e)\), the plaintext on the
selected branch \(b\) satisfies
\[
 d(\mathsf{dec}_0(sk,c),m_b)\leq e.
\]
Initialization establishes \(\Phi\), while encryption, evaluation, and the
deterministic prefix of decryption preserve it.
This is what makes the one-call KL calculation available after an arbitrary
adaptive history.

A machine-checked proof of \Cref{eq:central-game-hop} therefore needs a bridge
with three properties:
\begin{enumerate}
\item it retains conditional KL costs instead of converting each call to
statistical distance;
\item it exposes successive calls of an arbitrary adaptive oracle program; and
\item it threads \(\Phi\) through the exposed calls.
\end{enumerate}
\Cref{sec:pyth-toolkit} develops the Pythagorean judgment and invariant rules.
\Cref{sec:trace-compiler} verifies the selected-call compiler and proves the
generic local-to-adaptive theorem independently of FHE.  The instantiation in
\Cref{sec:closing-hop} combines them to derive the paid arrow in
\Cref{eq:conceptual-game-chain}.

\section{A Program Logic for Adaptive Games}
\label{sec:pyth-toolkit}

SSProve combines high-level algebraic reasoning about packages with a relational program logic for their low-level procedures~\cite{ssprove}.
Its central bridge turns per-procedure exact judgments into perfect indistinguishability of packages.
We retain that two-level architecture but add the quantitative judgments needed
by noise flooding: they carry conditional KL costs through program composition
and convert them to additive error only once.
This section develops those judgments independently of the FHE instantiation.
The verified compiler in \Cref{sec:trace-compiler} then lifts the resulting
one-call rule to complete adaptive oracle programs.

The development introduces no new programming language and takes no new logic
rules as axioms.
It defines semantic judgments over SSProve's existing semantics, then verifies the soundness of every rule below.
The Pythagorean judgment retains the conditional KL information required by
\Cref{sec:paid-hop}, while unary Hoare reasoning maintains the invariant from
\Cref{sec:adaptive-invariant}.
An additive-error judgment records the ordinary game-hop bound after the final
conversion.

We use SSProve's existing language of typed oracle programs, packages, heaps,
linking, and discrete subdistribution semantics; \Cref{app:ssprove-substrate}
recalls the precise fragment used here.

\subsection{Completed semantics and three judgments}

Because SSProve programs may fail, their denotations are subdistributions and
may have weight below one.
For
a subdistribution \(\mu\) on \(X\), define its completion on
\(X\cup\{\bot\}\) by
\[
 \completed{\mu}(x)=\mu(x),\qquad
 \completed{\mu}(\bot)=1-\sum_x\mu(x).
\]
Completion makes assertion failure observable.
Without this, it becomes tedious (and sometimes tricky) to write down couplings for subdistributions and then estimate the probabilities of postconditions.

The unary judgment
\[
 \vDash\mathsf{Hoare}\{P\}\ c\ \{Q\}
\]
is the standard Hoare triple over the SSProve code.
It states that whenever the precondition \(P\) is satisfied,
the postcondition \(Q\) is satisfied with probability one.

The additive-error judgment
\[
 \vDash\AEJ\{P\}\ c_L\approx_\delta c_R\ \{Q\}
\]
and its inference rules are heavily inspired by EasyCrypt style ``up-to-bad" arguments.
It states that, from any pair of initial configurations satisfying \(P\), the
completed outputs of \(c_L\) and \(c_R\) admit a coupling in which \(Q\) fails
with probability at most \(\delta\).
For equality postconditions, this is the ordinary statistical-distance bound.
Unlike EasyCrypt, SSProve allows its users to supply soundness proofs;
here we verify consequence, sequencing with loss \(\delta_1+\delta_2\),
triangle, projection, and direct-coupling rules.

We next present our Pythagorean judgment formally.

\begin{definition}[Pythagorean judgment]
\label{def:pyth-judgment}
Let \(c_L:X_L\to\code(Y)\) and \(c_R:X_R\to\code(Y)\), and let
\(s=(s_1,\ldots,s_k)\) be a nonempty tuple of real numbers.
For a configuration \(\rho=(x,m)\), write
\(\csemantic{c}_{\rho}=\completed{\semantic{c(x)}(m)}\), and let
\(\Omega=(Y\times\cM)\cup\{\bot\}\) be the completed output/heap space.
We write
\begin{equation}
 \vDash\Pyth\{P\}\ c_L\approx_s c_R\ \{Q\}
\end{equation}
if every \(s_i\) is nonnegative and, for every pair
\(\rho_L=(x_L,m_L)\) and \(\rho_R=(x_R,m_R)\) satisfying
\(P(\rho_L,\rho_R)\), there exist weight-one joint transcript distributions
\[
 \cP,\cQ\in\distr(\Omega^k)
\]
such that
\[
\begin{aligned}
 &\cP_k=\csemantic{c_L}_{\rho_L},
 \\
 &\cQ_k=\csemantic{c_R}_{\rho_R},
 \\
 &\forall i\in\{1,\ldots,k\},\ \forall a\in\Omega^{i-1}.\quad
   \mathsf{finiteKL}\!\left(
     \cP_i\mid a,\cQ_i\mid a\right),
 \\
 &\forall i\in\{1,\ldots,k\},\ \forall a\in\Omega^{i-1}.\quad
   \KL{\cP_i\mid a}{\cQ_i\mid a}\leq s_i,
 \\
 &\forall z\in\supp\semantic{c_L(x_L)}(m_L).\quad Q(z),
 \\
 &\forall z\in\supp\semantic{c_R(x_R)}(m_R).\quad Q(z).
\end{aligned}
\]
\end{definition}

This definition captures the Micciancio--Walter \cite{mw17} Pythagorean
prerequisites for KL divergence.
The transcript is semantic rather than a list of syntactic sampling sites: it may record a
whole program fragment as one coordinate, or expose multiple coordinates for
composition.

This judgment deliberately stores more information than an immediate statistical-distance bound.
Its sequence rule concatenates tuples (i.e. $s\doubleplus s'$) rather than adding square roots, allowing one nonlinear conversion after the full adaptive transcript has been constructed.

\subsection{Core proof rules}
\label{sec:logic-rules}

\Cref{fig:logic-rules} shows the rules that drive the generic development.
They are schematic and suppress most SSProve typing, full-mass, support, and
heap side conditions; the artifact contains their concrete statements.
For completeness, \Cref{fig:security-rule-inventory} in the appendix collects
every rule used by the noise flooding security analysis, including the Hoare, additive-error, derived Gaussian, and
compiler rules omitted here.

\begin{figure*}[t]
\centering
\small
\begin{mathpar}
\inferrule*[right=\textsc{KL-Sample}]
  {\mathsf{finiteKL}(D_L,D_R) \\
   \KL{D_L}{D_R}\leq\epsilon \\
   \mathsf{weight}(D_L)=\mathsf{weight}(D_R)=1 \\
   \forall x\in\supp(D_L)\cup\supp(D_R),\ Q(x)}
  {\vDash\Pyth\{P\}\
     x\leftarrow D_L
     \approx_{[\epsilon]}
     x\leftarrow D_R\ \{Q\}}

\inferrule*[right=\textsc{Pyth-Refl}] {P\Longrightarrow x_L=x_R\wedge m_L=m_R \\ 0\leq s_i\ \text{for every }i} {\vDash\Pyth\{P\}\ c\approx_s c\ \{Q\}}

\inferrule*[right=\textsc{Pyth-Seq}] {\vDash\Pyth\{P\}\ c_L\approx_s c_R\ \{M\} \\ \vDash\Pyth\{M^{=}\}\ k_L\approx_t k_R\ \{Q\}} {\vDash\Pyth\{P\}\ (x\leftarrow c_L;k_L(x)) \approx_{s\doubleplus t} (x\leftarrow c_R;k_R(x))\ \{Q\}}

\inferrule*[right=\textsc{Shared-Prefix}] {\vDash\mathsf{Hoare}\{P\}\ c\ \{M\} \\ \vDash\Pyth\{M^{=}\}\ k_L\approx_s k_R\ \{Q\}} {\vDash\Pyth\{P\}\ (x\leftarrow c;k_L(x)) \approx_{[0]\doubleplus s} (x\leftarrow c;k_R(x))\ \{Q\}}

\inferrule*[right=\textsc{Micciancio--Walter}]
  {\vDash\Pyth\{P\}\ c_L\approx_s c_R\ \{Q\}}
  {\vDash\AEJ\{P\}\ c_L
     \approx_{\sqrt{\norm{s}_1/2}}c_R\ \{o_L=o_R\}}
\end{mathpar}
\caption{Representative checked rules.
Here \(M^{=}\) requires equal
intermediate values and heaps satisfying \(M\), and \(\doubleplus\) denotes
tuple concatenation.
The Rocq theorems additionally expose further typing,
support, invariant, and nonnegativity premises.}
\label{fig:logic-rules}
\end{figure*}

\textsc{KL-Sample} embeds a distribution-level KL bound as a program judgment.
\textsc{Pyth-Refl} relates a program to itself at zero cost.
\textsc{Pyth-Seq} is the essential rule: its proof constructs a joint transcript for semantic bind and proves conditional bounds for both valid and zero-mass prefixes.
The development also contains a shared-continuation rule, consequence rules,
variants for fully linked programs, and the additive-error rules mentioned
above.

The mixed Hoare/Pythagorean rules have been useful in practice.
Deterministic prefixes, table operations, other oracle calls, and postprocessing are proved to preserve an invariant with the unary Hoare logic and inserted as zero-cost coordinates in the Pythagorean judgment.
Only the sampling fragment whose distribution changes consumes KL budget.
In the noise-flooding application of \Cref{sec:closing-hop}, this decomposition
produces the tuple \((0,\epsilon_{\mathsf{nf}},0)\).

\subsection{The Micciancio-Walter Probability Lemma}

The \textsc{Micciancio--Walter} rule in \Cref{fig:logic-rules} rests on the
following distribution-level theorem.
Let \(\mathcal P,\mathcal Q\) be joint distributions on \(\Omega^k\), and write \(\mathcal P_i\mid a\) for coordinate \(i\) conditioned on prefix \(a\in\Omega^{i-1}\).

\begin{theorem}[Conditional-coordinate preservation]
\label{thm:pyth-probability}
Suppose \(\mathcal P,\mathcal Q\) have weight one and, for every coordinate
\(i\) and prefix \(a\), the conditional KL divergence is finite and
\[
 \KL{\mathcal P_i\mid a}{\mathcal Q_i\mid a}\leq s_i,
 \qquad s_i\geq0.
\]
Then their final-coordinate marginals satisfy
\[
 \TVD(\mathcal{P}_k,
      \mathcal{Q}_k)
 \leq \sqrt{\frac{\sum_{i=1}^{k}s_i}{2}}.
\]
\end{theorem}

At paper level, the proof expands joint mass as prefix mass times conditional mass and groups the log likelihood ratio by coordinate.
The resulting chain bound gives \(\KL{\mathcal P}{\mathcal Q}\leq\sum_i s_i\).
Pinsker bounds the transcript distance, and data processing gives the displayed final-marginal bound.
The Rocq proof makes explicit absolute continuity, conditional zero-mass prefixes, summability of positive and negative parts, and exchanges between finite coordinate sums and countable support sums.
We further discuss this proof in the Appendix \Cref{sec:verified-kl-analysis}.

Our \textsc{Micciancio--Walter} rule instantiates this theorem with the transcripts from \Cref{def:pyth-judgment}, then uses maximal coupling to obtain the additive-error equality judgment.
As with all of our inference rules, we verify it to be sound with respect to SSProve's program semantics.

\subsection{A derived discrete-Gaussian rule}

For center \(\mu\in\bbZ\) and \(\sigma>0\), we denote the discrete Gaussian distribution by
\[
 D_{\bbZ,\mu,\sigma^2}(x)
 =\frac{\exp(-(x-\mu)^2/(2\sigma^2))}
        {\sum_{z\in\bbZ}\exp(-(z-\mu)^2/(2\sigma^2))}.
\]
The development proves normalization, translation invariance of the
normalizer, summability, the centered first moment, full support, and the
equal-variance KL identity in \Cref{eq:dg-kl}.
The calculation expands the
logarithm of the ratio; equal variances cancel both normalizers, and symmetry
of the centered distribution cancels the linear moment.
We again defer the proof details until \Cref{sec:verified-kl-analysis}.

These facts yield a reusable program rule for vector sampling.
If
\(u,v\in\bbZ^n\), \(\sigma>0\), and every coordinate satisfies
\[
 \frac{(v_i-u_i)^2}{2\sigma^2}\leq\epsilon,
\]
then the $n$-fold product sampler satisfies a Pythagorean
judgment with budget \([n\epsilon]\).
The analysis first constructs
an \(n\)-coordinate transcript, proves the scalar bound coordinatewise, and
then transports and aggregates it into the program-level sampling rule.
The
noise-flooding proof instantiates \(\epsilon=1/(2\gamma^2)\), obtaining
\(\epsilon_{\mathsf{nf}}=n/(2\gamma^2)\) without introducing any unverified assumptions.

\subsection{Why the extra judgment is necessary}

The additive-error layer remains essential for exact hops, up-to-bad arguments, and the final triangle calculation.
Its ordinary sequence rule, however, maps \(\delta_1,\delta_2\) to \(\delta_1+\delta_2\).
Applying Pinsker inside each decrypt proof would therefore pay \(q\sqrt{\epsilon_{\mathsf{nf}}/2}\).
The Pythagorean layer delays that nonlinear conversion and obtains \(\sqrt{q\epsilon_{\mathsf{nf}}/2}\).
It provides a verified relational interface for carrying parameter-critical KL
information across stateful program composition.

\section{A Verified Compiler for Adaptive Oracle Programs}
\label{sec:trace-compiler}

The rules above compose a sequence of sampling operations.
An oracle program does not provide such a sequence: the location and arguments of its next call may depend on previous answers.
Our compiler is the verified adapter between these views.
It exposes a bounded number of calls to one selected operation without
restricting the surrounding program's control flow, and it comes with both a
same-package correctness theorem and a generic quantitative replacement theorem.

\subsection{A trace-based selected-call transformation}

Given a program \(A\) and a chosen operation \(p\), a paper proof can simply
say: run \(A\) until its next call to \(p\), answer that call, and resume.
In SSProve, the code after a call is represented by a
Rocq function from the answer to the rest of the program.  Such arbitrary
functions cannot be serialized into SSProve program data.  The compiler therefore records a
trace of the intervening calls, heap operations, and samples.
When it reaches \(p\), it returns the query without executing the call.  The
driver obtains an answer from a chosen implementation \(P_s\), adds it to the
trace, replays the recorded results through \(A\) to recover the code after the
call, and repeats this process up to \(q\) times.

Write \(\mathsf{Expose}_q(A,p;P_s)\) for the raw-code transformation that
exposes the first \(q\) calls made by \(A\) to operation \(p\) and resolves
them with \(P_s\).  Its factoring phase is trace-valued, but replay restores
the residual program, so the complete transformation has the same result type
as \(A\).  Write \(\mathsf{Compile}_q(A,p;P_s,P_r)\) for
\(\mathsf{Expose}_q(A,p;P_s)\) linked with residual package \(P_r\).

\subsection{Same-package correctness}

The transformation inserts trace bookkeeping, so before linking its syntax
differs from that of \(A\).  If both the exposed calls and the residual calls
are implemented by \(P\), however, each call receives the same answer as in
the original linked program, and trace replay reconstructs the same remaining
code.  The compiler therefore preserves behavior in the same-package case.

\begin{theorem}[Same-package compilation]
\label{thm:compile-same}
Subject to code and package validity and the presence of the selected
operation, for every initial heap,
\[
 \mathsf{Compile}_q(A,p;P,P)
 \quad\text{and}\quad
 \mathsf{Link}(A,P)
\]
have identical SSProve output/heap semantics.
\end{theorem}

The proof is an induction over the selected-call traversal.  At each step, the
recorded result is supplied to the corresponding continuation of the original
program, while validity rules out malformed traces.  Consequently, the
compiled and original programs induce the same output--heap distribution.

\subsection{A generic local-to-adaptive replacement theorem}

The following theorem is the main program-logic bridge.

\begin{theorem}[Replace $q$ calls]
\label{thm:compile-replace}
Let \(P_I\) synchronize equal operation inputs and equal heaps satisfying
\(I\).  Suppose two implementations \(P,P'\) of an operation \(p\) satisfy
the one-call Pythagorean judgment
\[
 \vDash\Pyth\{P_I\}\ P.p\approx_s P'.p\ \{I\}
\]
for a nonnegative tuple \(s\), preserve \(I\), and meet the compiler's typing,
footprint, and memory-separation conditions.
For every concrete SSProve program \(A\), replacing its first
\(q\) selected calls yields
\[
\begin{aligned}
 &\vDash\AEJ\{P_I\}\ \mathsf{Compile}_q(A,p;P,P)\\
 &\quad\approx_{\sqrt{q\norm{s}_1/2}}
   \mathsf{Compile}_q(A,p;P',P)\ \{o_L=o_R\},
\end{aligned}
\]
\end{theorem}
\begin{proof}[Proof sketch]
We first show that \(\mathsf{Compile}_q(A,p;P,P)\) and
\(\mathsf{Compile}_q(A,p;P',P)\) satisfy the expected Pythagorean judgment with
\(q\) concatenated copies of the tuple \(s\), by induction on \(q\) and the
Pythagorean sequence rule.
Each exposed call contributes \(s\) and preserves the invariant \(I\).

After \(q\) calls, the accumulated tuple has norm
\(q\norm{s}_1\); \textsc{Micciancio--Walter} is applied once, giving the
desired additive error bound.
\end{proof}

The theorem is generic in the program, packages, selected oracle, invariant, and budget tuple.
The noise-flooding game supplies only the one-call premise proved in
\Cref{sec:closing-hop}; the program logic and compiler then handle adversarial
control flow.
Combined with \Cref{thm:compile-same}, this is our quantitative counterpart of SSProve's local relational reasoning for package indistinguishability.

\subsection{Instantiating the compiler for noise flooding}
\label{sec:closing-hop}

We now instantiate the generic machinery to prove the missing arrow in
\Cref{eq:conceptual-game-chain}.
Let \(C_{\cA}^{\mathsf{open}}\) be the adversary-facing IND-CPAD game code
before its oracle implementations are linked; its encryption, evaluation, and
decryption calls remain unresolved.
Let \(P_{\mathsf{real}}\) implement the real encryption, evaluation, and
flooded-decryption oracles, and let \(P_{\mathsf{sim}}\) differ only by
flooding around the common recorded plaintext when the row's two plaintexts
agree.

The selected decryption body factors into a deterministic prefix, one Gaussian
sample, and deterministic postprocessing.
The invariant \(\Phi\) from \Cref{sec:adaptive-invariant} gives
\(d(\mathsf{dec}_0(sk,c),m)\leq e\) on every reachable row whose recorded
plaintexts agree.
Translation of charts by \Cref{eq:charts} puts the two answer distributions
over the same coordinate space.
Applying \Cref{eq:dg-kl} in \(n\) coordinates with
\(\sigma_e=\max\{1,e\}\gamma\) yields the following local judgment.

\begin{lemma}[One-call replacement]
\label{lem:one-call}
For any decryption input and paired heaps satisfying \(\Phi\), the real and
simulated decryption bodies satisfy a Pythagorean relational judgment with
budget tuple
\[
 (0,\epsilon_{\mathsf{nf}},0),
 \qquad
 \epsilon_{\mathsf{nf}}=\frac{n}{2\gamma^2},
\]
and restore \(\Phi\).
\end{lemma}

The zero entries record the shared prefix and postprocessing.
Branches with differing recorded plaintexts, as well as assertion-failing
branches, agree exactly.
Applying \Cref{thm:compile-replace} to \Cref{lem:one-call} concatenates one
such tuple for each of the first \(q\) decryption calls, preserving all
dependencies between queries.
The \textsc{Micciancio--Walter} rule is then invoked once, after the accumulated
KL budget is \(q\epsilon_{\mathsf{nf}}\), proving
\[
 \TVD(G_0,G_1)\leq\sqrt{q\epsilon_{\mathsf{nf}}/2}.
\]
Data processing projects completed output/heap pairs to the returned Boolean.
This proves \Cref{eq:central-game-hop}.

\smallskip\noindent\emph{From games to program representations.}
The \(G_i\) in \Cref{eq:conceptual-game-chain} are complete experiments: the
adversary has already been linked to its oracle implementations, so no
unresolved calls remain.
This is the usual cryptographic view of a game, but it hides the program
structure needed by \Cref{thm:compile-replace}.
The formal proof therefore makes the linking and call-compilation steps
explicit.
For selected-call implementation \(P_s\) and residual implementation \(P_r\),
write
\[
 H(P_s,P_r)=
 \mathsf{Compile}_{q}
   (C_{\cA}^{\mathsf{open}},\mathsf{ODec};P_s,P_r).
\]
Then
\begin{align*}
 H_0 &= G_0,\\
 H_1 &= \mathsf{Link}(C_{\cA}^{\mathsf{open}},P_{\mathsf{real}}),\\
 H_2 &= H(P_{\mathsf{real}},P_{\mathsf{real}}),\\
 H_3 &= H(P_{\mathsf{sim}},P_{\mathsf{real}}),\\
 H_4 &= H(P_{\mathsf{sim}},P_{\mathsf{sim}}),\\
 H_5 &= \mathsf{Link}(C_{\cA}^{\mathsf{open}},P_{\mathsf{sim}}),\\
 H_6 &= G_2.
\end{align*}
This is an implementation-level factorization of the two arrows in
\Cref{eq:conceptual-game-chain}, not seven cryptographic hybrids:
\(H_0,H_1,H_2\) are real-oracle representations of \(G_0\);
\(H_3,H_4,H_5\) are simulated-oracle representations of \(G_1\); and
\(H_6\) is the reduction game \(G_2\).
Rocq verifies
\begin{equation}
\label{eq:hybrid-chain}
 H_0\xrightarrow{0}H_1\xrightarrow{0}H_2
 \xrightarrow{\sqrt{q\epsilon_{\mathsf{nf}}/2}}H_3
 \xrightarrow{0}H_4\xrightarrow{0}H_5\xrightarrow{0}H_6.
\end{equation}
Here \(H_0=H_1\) unfolds package linking, and \(H_1=H_2\) is same-package
compiler correctness.
The single paid step \(H_2\to H_3\) is the application above.
For \(H_3=H_4\), the shared counter ensures that any residual decryption call
after the \(q\) selected entries assertion-fails on both sides.
The equality \(H_4=H_5\) removes the compiler using
\Cref{thm:compile-same}.
Finally, \(H_5\to H_6\) relates the simulated IND-CPAD presentation to the
concrete IND-CPA reduction by an exact heap invariant; the artifact factors
this relation through auxiliary linked programs.

Collapsing the zero-cost equalities in \Cref{eq:hybrid-chain} recovers the
three-game chain in \Cref{eq:conceptual-game-chain}; only \(H_2\to H_3\)
changes a distribution.
Triangle composition of \Cref{eq:hybrid-chain}, followed by the assumed
IND-CPA bound at \(H_6\), proves \Cref{thm:main}.

\section{Mechanization and Trusted Base}
\label{sec:artifact}

The artifact contains 27,987 lines of Rocq.
Since our Rocq development is heavily AI-assisted,
this line count does not accurately reflect our engineering effort.
The overview in \Cref{fig:theory-map} separates the application-specific
security argument from the reusable probability and program-logic results.

The mechanization contains a few layers.  The scheme layer
defines encryption games and the noise-flooded transform.  The
probability layer develops facts about countable discrete distributions,
without reference to programs.  The logic layer lifts those facts to SSProve
code.  The compiler layer implements the trace transformation and derives the
generic adaptive replacement rule.  Finally, the security
layer instantiates the reusable logic and compiler with the decryption
invariant, constructs the IND-CPA reduction, and composes the exact and
quantitative game transitions.
The top-level result is a functor over precisely the scheme, chart,
probability-one correctness, IND-CPA security, and flooding-parameter
interfaces described in \Cref{sec:construction}; probability lemmas and
program rules are proved dependencies.

\subsection{Formalization of judgments and rules}

We follow SSProve's convention and engineering framework: pure computation is
embedded in Rocq, while probabilistic and stateful effects are represented by
SSProve's free-monadic program syntax, packages, heaps, and linking operations
\cite{ssprove}.  The new logic does not extend Rocq with a primitive judgment
or a collection of trusted proof rules.
Judgments are formalized as ordinary Rocq propositions over the denotation \texttt{Pr\_code},
and inference rules are lemmas which prove such judgments.

For example, the exact definition of the additive-error judgment is shown
below.  Its types range over SSProve programs and heaps; \texttt{complete}
adds an explicit outcome for missing subdistribution mass.

\begingroup
\setcoqmonofont
\begin{minted}{coq}
Definition additiveErrorJudgment
  {inL_t inR_t outL_t outR_t : ord_choiceType}
  (progL : inL_t -> raw_code outL_t)
  (progR : inR_t -> raw_code outR_t)
  (pre : pred ((inL_t * heap) * (inR_t * heap)))
  (post : pred (option (outL_t * heap) * option (outR_t * heap)))
  (ε : R) : Prop :=
  0 <= ε /\
  ∀ memL memR xL xR, pre ((xL, memL), (xR, memR)) →
    let out1 := Pr_code (progL xL) memL in
    let out2 := Pr_code (progR xR) memR in
    ∃ d, coupling d (complete out1) (complete out2) ∧
      \P_[ d ] post >= 1 - ε.
\end{minted}
\endgroup

Thus, for every related pair of inputs and initial heaps, the proposition
requires a coupling of the completed output-and-heap distributions in which
the postcondition holds with probability at least \(1-\varepsilon\).
The Pythagorean judgment is encoded in the same style.  For each related input
pair it existentially quantifies two full transcript distributions, requires
their final marginals to equal the two completed program denotations, and
requires their supports to satisfy the postcondition.  Its predicate
\texttt{pythDist} additionally imposes full mass and the coordinatewise
conditional-KL bounds from \Cref{def:pyth-judgment}.
To place every result type in one transcript space, the Rocq definition
injectively encodes typed output/heap pairs in \texttt{nat * heap} and
represents \(\bot\) by \texttt{None}.

The rules in \Cref{fig:logic-rules} are proved implications between these
definitions.  For instance, the source-level statement of the final
KL-to-additive-error bridge is:

\begingroup
\setcoqmonofont
\begin{minted}{coq}
Lemma MicciancioWalterRule
  {ℓ : nat}
  {inL_t inR_t out_t : choice_type}
  (progL : inL_t -> raw_code out_t)
  (progR : inR_t -> raw_code out_t)
  (pre : pred ((inL_t * heap) * (inR_t * heap)))
  (post : pred (out_t * heap))
  (s : (ℓ.+1).-tuple R) :
  ⊨Pyth ⦃ pre ⦄ progL ≈( s ) progR ⦃ post ⦄ ->
  let delta := pythagorean_tv_bound s in
  ⊨AE ⦃ pre ⦄ progL ≈( delta ) progR ⦃
    fun outs =>
    let '(outL, outR) := outs in
    outL == outR ⦄.
\end{minted}
\endgroup

Here \texttt{pythagorean\_tv\_bound s} is
\(\sqrt{\sum_i s_i/2}\).  Its proof applies the verified probability theorem
to the transcript witnesses, projects to their final marginals, and invokes
the completed-distribution coupling theorem.  The sampling, sequencing,
mixed Hoare/Pythagorean, and compiler rules have the same status: each is a
Rocq lemma proved from the program semantics.  Consequently, a Rocq kernel check of
the security theorem also checks the soundness chain from the distribution
lemmas through the program judgments.

\subsection{Compiler implementation and correctness}

\Cref{fig:compiler-code} shows the traversal over SSProve's \texttt{raw\_code}
syntax.  Its result pairs the trace with a status: \texttt{inl} carries the
encoded input of a selected call, whereas \texttt{inr} carries the final output
of a program that terminated first.  A nonselected call is executed, its
result is appended to the trace, and traversal continues; heap operations and
samples are handled analogously.

\begin{figure}[t]
\begingroup
\setcoqmonofont
\begin{minted}{coq}
Fixpoint run_until_next_call_aux {T : choice_type} (prog : raw_code T) (fn : ident) (trace : trace_t) :
  raw_code suspended_program :=
  match prog with
  | ret v => ret (inr v, pack_trace trace)
  | opr o x k =>
    let '(f, _) := o in
    if f == fn then
      ret (inl (pickle x), pack_trace trace)
    else (
      y ← op o ⋅ x ;;
      run_until_next_call_aux (k y) fn (rcons trace (call_entry (pickle y)))
    )
\end{minted}
\endgroup
\caption{Scanning for the next call to \texttt{fn}: \texttt{inl} reports its
encoded input, while \texttt{inr} reports program termination.}
\label{fig:compiler-code}
\end{figure}

The outer driver answers a discovered call, records the answer, and repeats.
After the requested calls, trace replay reconstructs the remaining code, so
the complete transformation has the original result type.  The same-package
correctness result is \Cref{thm:compile-same}, whose Rocq theorem is
\texttt{compile\_calls\_correct}.  It proves equality of program semantics by
induction over traversal and trace replay.

The separate quantitative result is \Cref{thm:compile-replace}.  Its call
induction lifts the one-call Pythagorean judgment through the trace
transformation and accumulates \(q\) copies of its KL-cost tuple;
\texttt{compileRule} then performs the single final Micciancio--Walter
conversion.  The \textsc{Compile} rule in
\Cref{fig:security-rule-inventory} records its complete security-facing
premises.  Thus both semantic compiler correctness and adaptive quantitative
replacement are checked consequences of SSProve syntax and semantics.

\subsection{Probability formalization}

The probability development works over MathComp's countable discrete
distributions, rather than restricting the Gaussian argument to finite
support.  It defines KL finiteness as absolute continuity together with
summability of the KL integrand, proves Pinsker's inequality, and proves the
conditional-coordinate chain bound used in
\Cref{thm:pyth-probability}.  It also constructs completion and maximal
couplings, so that a distribution-level total-variation bound yields the
program-level proposition above.

The distributions used here, notably the discrete Gaussians, have countably
infinite support.  Consequently, the KL chain rule involves infinite series.
A paper proof may leave their convergence implicit; the formalization must
establish it before decomposing transcript KL into conditional KL costs.  It
also treats histories of probability zero explicitly.  These results yield the
chain bound used in the square-root composition argument.

For the one-call bound, the development also proves that the discrete Gaussian
is normalized and has full support, that its normalizer is invariant under
integer shifts, and that its centered first moment is zero.  These facts yield
the exact equal-variance KL identity in \Cref{eq:dg-kl} and discharge the
Gaussian sampling rule.  Further theorem statements and proof structure
appear in \Cref{sec:verified-kl-analysis}.

\subsection{The checked theorem}

At the application boundary, the functor defines the concrete reduction
\(\cB_{\cA,q}\) and makes the assumed IND-CPA bound and the statistical loss
explicit:

\begingroup
\setcoqmonofont
\begin{minted}{coq}
  Definition security_bound (A : nom_package) (max_queries : nat) :=
    let B := ind_cpa_reduction A max_queries in
    IndCpaSecurity.security_bound B +
      security_loss dim max_queries gaussian_width_multiplier.
\end{minted}
\endgroup

After proving that the constructed reduction has the IND-CPA adversary
interface and deriving the completed-output inequality, the development
exports the following theorem.

\begingroup
\setcoqmonofont
\begin{minted}{coq}
  Theorem is_secure (A : nom_package) max_queries :
    Package IndCpaDSim.IndCpadAdv_import IndCpaDSim.IndCpadAdv_export A ->
    IndCpadGame.winning_probability max_queries A <=
    security_bound A max_queries.
\end{minted}
\endgroup

The premise says that \(A\) is a well-typed package with the required IND-CPAD
adversary interface.  Unfolding \texttt{security\_bound} and
\texttt{security\_loss} gives exactly \Cref{thm:main}; in particular, the
right-hand term is the assumed IND-CPA winning-probability bound for the
constructed reduction plus \(\sqrt{qn}/(2\gamma)\).  This source statement also
makes clear where the theorem boundary lies: scheme structure, charts,
probability-one correctness, IND-CPA security, and positivity of \(\gamma\)
enter through the enclosing functor, whereas reduction validity, the
discrete-Gaussian calculation, compiler correctness, and logic soundness have
already been proved.

\subsection{Reproducibility}

\ifdefined\CameraReady
The complete artifact is available at
\url{https://github.com/ethanlee515/Mending}.
\fi
The complete artifact checks with Rocq 9.0.1, \texttt{rocq-ssprove}
development version at commit \texttt{c6d7d4bc3a}, MathComp Analysis 1.16.0,
and MathComp algebra tactics 1.2.7.  In a fresh switch on the authors' machine,
a clean build took about six minutes with four parallel jobs; this is a single
observation for reviewer planning, not a performance claim.  The manuscript's
Rocq listings are refreshed from the checked sources during its build, so the
displayed propositions and theorem statements do not silently diverge from
the artifact.

\subsection{Trusted computing base}
\label{sec:tcb}

The trusted base includes the Rocq kernel and standard library; MathComp,
MathComp Analysis, and their real-number/classical infrastructure; and
SSProve's program, heap, package, linking, and subdistribution semantics.
The scheme, chart, correctness, parameter, and IND-CPA interfaces from
\Cref{sec:construction} are theorem parameters and must be discharged by a
concrete instantiation.

The new KL, Gaussian, relational, compiler, and game-hop results contain no
local \texttt{Admitted}.  The installed dependency graph nevertheless reports
MathComp Analysis's
\texttt{\_\_admitted\_\_interchange\_psum}: current SSProve
distribution semantics are built on these MathComp definitions and reach the
upstream unproven lemma ~\cite{ssprove}.  This is an inherited assumption, not an
unproven axiom introduced by our security development.

We independently discharge the mathematical obligation.  The artifact proves
the same statement as \texttt{interchange\_psum\_proved};
\texttt{Print Assumptions} for this replacement reports only the usual
extensionality and choice principles and not the admitted upstream lemma.  The
proof has since been merged to MathComp Analysis's development branch,
\ifdefined\CameraReady
via \href{https://github.com/math-comp/analysis/pull/2007}%
{\texttt{math-comp/analysis\#2007}},
\fi
for inclusion in the 1.17.0 release.
Before this patch eventually reaches SSProve,
\texttt{Print Assumptions} for the top-level theorem
will continue to display the admitted lemma.

Besides library and classical assumptions, \texttt{Print Assumptions} for an
instantiated theorem reports the concrete scheme, chart, correctness, and
IND-CPA interfaces, together with positivity of the discrete-Gaussian width
multiplier \(\gamma\).  It does not report any program-logic rule, compiler
correctness theorem, or discrete-Gaussian KL identity: these are verified
inside the development.

Selected definitions of the scheme interface, games, transform,
and reduction remain in \Cref{app:formal-interface-excerpts}.  Together with
the source-level theorem above, they expose the small specification surface on
which a cryptographic audit should concentrate; the Rocq kernel checks that the
remaining proof terms connect that surface by the argument in
\Cref{sec:verified-reduction}.  As in SSProve and EasyCrypt, running time is not
part of the program semantics and must be inspected separately
\cite{easycrypt,ssprove}.  Here the
reduction runs the adversary once, forwards its oracle calls, and adds table
operations and discrete-Gaussian sampling.

\section{Related Work}
\label{sec:related-work}

\emph{Approximate-FHE attacks and defenses.} Li and Micciancio identified the secret-dependent approximation-error channel in CKKS-style decryption~\cite{lm}.
LMSS introduced the IND-CPAD formulation and used differential-privacy and noise-flooding techniques to recover security under explicit parameter conditions~\cite{lmss}.
Subsequent work has sharpened the practical precision tradeoff~\cite{costache2023precision} and revisited concrete noise-flooding parameters~\cite{bergamaschi2025revisiting}.
Attacks against exact FHE at aggressive parameters show that decryption-query security is not confined to approximate schemes~\cite{cheon2024attacks}.
Our contribution is neither a new attack nor a better concrete parameter set: it machine-checks an existing reduction and the square-root composition mechanism on which such parameters rely.

\emph{Mechanized game-based cryptography.} SSProve combines state-separating package algebra with a probabilistic relational program logic in Rocq and validates the framework on encryption, KEM--DEM, and sigma-protocol case studies~\cite{ssprove}.
Its central bridge lifts invariant-preserving exact judgments for individual procedures to perfect indistinguishability of packages.
Our result follows the same security-facing pattern at a different quantitative boundary: it lifts a conditional-KL judgment for one oracle implementation to additive indistinguishability of a complete adaptive program.
We retain SSProve's syntax, package algebra, heaps, linking, and subdistribution semantics, while adding judgments and a selected-call compiler.
Recent Nominal-SSProve work makes package state separation intrinsic and reports errors found in an ElGamal mechanization~\cite{nominal-ssprove}; a further development mechanizes query-counting, multi-instance, and nested hybrid arguments~\cite{nested-hybrids}.
Those results address exact package composition and general hybrid structure.
Our compiler instead preserves an adaptive sequence of conditional KL costs and derives a sublinear statistical-distance loss.

EasyCrypt is another popular framework for game-based cryptographic proofs.
Adversaries are treated as opaque modules, and reflection techniques \cite{reflection} can be used to extract the adversary's semantics.
Our abstraction boundary differs because SSProve adversaries are concrete program syntax: the verified trace compiler factors selected calls before the replacement theorem is applied.
It is not immediately clear whether EasyCrypt's reflection machinery suffices to formalize the FHE sub-linear security bound.
We leave this matter as another interesting open question.

\emph{Approximate relational logics.} Approximate relational Hoare logics and approximate liftings were developed largely for differential privacy~\cite{aprhl,aprhl-couplings}.
They support rich \((\varepsilon,\delta)\) reasoning, but their standard sequential rules add privacy or additive-error parameters.
One can instantiate general \(f\)-divergence machinery with KL divergence.
The issue here is where the nonlinear conversion occurs: applying Pinsker at each command boundary loses linearly.
Our semantic transcript judgment exposes conditional KL budgets to its sequence rule, concatenates them, and converts only at the final marginal.
The compiler then makes that rule applicable to adaptive oracle syntax.
The resulting logic is deliberately specialized to completed discrete semantics and the game-based additive conclusion used by noise flooding; it is not a general replacement for differential-privacy logics.

\section{Conclusion and Future Works}
\label{sec:limitations}

Our formalization deliberately factors correctness failures from the
quantitative noise-flooding argument.
Its probability-one hypotheses prove the good-execution case.
For an imperfectly correct FHE such as CKKS, one more standard up-to-bad game hop must be formalized before instantiation, and the probability of generating ``bad" key pairs and encryptions must be carefully accounted for.

The chart interface from \Cref{eq:charts} also remains abstract.
A concrete CKKS instantiation must connect its plaintext ring and norm to those
distance-and-translation laws.
Keeping that obligation separate makes the generic reduction easier to audit,
but does not discharge it.

Modulo the above proof engineering items, we have formally verified the central LMSS game hop \cite{lmss}.
That is, an IND-CPA and approximately correct FHE scheme can be made IND-CPAD through noise flooding.
The natural next step then, is to verify that the CKKS scheme is indeed both IND-CPA and approximately correct.
That will certainly come with its own set of other interesting challenges.
For instance, such an effort will no doubt contend with the use of randomized versus deterministic rounding, as well as the (difficult to formalize) heuristics involved in the deterministic case as found ``in the wilds".

\label{page:main-matter-end}

\section*{Acknowledgments}

Y.L., J.L., X.W. are partially supported by the National Science Foundation grant CCF-1942837 (CAREER), CCF-2330974, and a Sloan Research Fellowship. A.C. acknowledges support from the National Science Foundation grant CCF-1813814, from the AFOSR under Award Number FA9550-20-1-0108 and from the Quantum Advantage Pathfinder project.

OpenAI Codex was used to draft prose throughout this manuscript.
The Rocq development was also substantially AI-assisted.
Both have been audited and revised by the authors accordingly.

\bibliographystyle{IEEEtran} \bibliography{reference}

\appendix

\section{Security Games and Reduction Pseudocode}
\label{app:security-games}

This appendix gives the full security experiments used in the paper and
expands the reduction from \Cref{sec:verified-reduction} into
executable-style pseudocode.

\subsection{The IND-CPA experiment}
\label{app:ind-cpa-game}

IND-CPA is the standard confidentiality notion for public-key encryption.
We use its multi-query left-or-right formulation: the challenger samples one
hidden bit \(b\), and every encryption query selects its plaintext with that
same bit.
The evaluation key is public, so the adversary can evaluate ciphertexts
locally without a challenger oracle.

\begin{algorithm}
\caption{Multi-query left-or-right IND-CPA experiment}
\label{alg:ind-cpa-game}
\begin{algorithmic}
  \Procedure{Game}{$\cB$}
    \State $b \gets \{0,1\}$
    \State $(pk,evk,sk) \gets KeyGen()$
    \State $b^* \gets \cB^{OEnc_{\mathsf{CPA}}}(pk,evk)$
    \State \Return $(b^*=b)$
  \EndProcedure
  \State
  \Procedure{$OEnc_{\mathsf{CPA}}$}{$m_0,m_1$}
    \State $c \gets Enc(pk,m_b)$
    \State \Return $c$
  \EndProcedure
\end{algorithmic}
\end{algorithm}

The adversary may choose several encryption queries adaptively.
We state security in winning-probability form: the scheme's IND-CPA assumption
supplies a bound
\[
 \WinProb[\mathsf{IND\text{-}CPA}_{S}(\cB)]
 \leq \beta_{\mathsf{CPA}}(\cB)
\]
for every adversary \(\cB\) with the required oracle interface.
Equivalently, if advantage means excess winning probability over random
guessing, then
\(\beta_{\mathsf{CPA}}(\cB)=\tfrac12+\mathsf{Adv}_{\mathsf{CPA}}(\cB)\).

\subsection{The IND-CPAD experiment}
\label{app:ind-cpad-game}

Li and Micciancio introduced IND-CPA with decryption, or IND-CPAD, to model
attacks that exploit released approximate decryptions~\cite{lm}.
The game retains the hidden bit, encryption oracle, and final guess of
IND-CPA, and adds homomorphic evaluation and restricted decryption.
Its challenger maintains a table
\[
 T=[(m_0,m_1,c),\ldots].
\]
Each row records the actual ciphertext \(c\) together with the plaintext it
would represent in the left challenge world and in the right challenge world.
Only the ciphertext for the world selected by \(b\) is sampled.

\begin{algorithm}
\caption{IND-CPAD Challenger Skeleton}
\label{alg:ind-cpad-game}
Ingredients:
\begin{itemize}
  \item An approximate homomorphic encryption scheme
    $S=(KeyGen,Enc,Eval,Dec)$.
  \item A decrypt-query bound $q\in\bbN$.
\end{itemize}
Global challenger state:
\begin{itemize}
  \item keys $pk,evk,sk$;
  \item hidden challenge bit $b\in\{0,1\}$;
  \item table $T\in(M\times M\times C)^*$ of challenge plaintext pairs and
    produced ciphertexts;
  \item decrypt counter $d$.
\end{itemize}
\begin{algorithmic}
  \Procedure{Game}{$\cA$}
    \State $b \gets \{0,1\}$
    \State $(pk,evk,sk) \gets KeyGen()$
    \State $T \gets [\,]$
    \State $d \gets 0$
    \State $b^* \gets \cA^{OEnc,OEval,ODec}(pk,evk)$
    \State \Return $(b^*=b)$
  \EndProcedure
\end{algorithmic}
\end{algorithm}

The encryption oracle appends a new row.
The evaluation oracle applies the same gate separately to the recorded
left-world and right-world plaintexts, then stores both results alongside the
evaluated ciphertext.
The decryption oracle returns a plaintext only if the two recorded plaintexts
agree.
This restriction is what makes decryption compatible with a left-or-right
game: if the two plaintexts differed, releasing the decryption could reveal
the hidden bit immediately.

\begin{algorithm}
\caption{IND-CPAD Oracles}
\label{alg:ind-cpad-oracles}
The oracles share the challenger state from \Cref{alg:ind-cpad-game}.
\begin{algorithmic}
  \Procedure{OEnc}{$m_0,m_1$}
    \State $c \gets Enc(pk,m_b)$
    \State $T \gets T \doubleplus [(m_0,m_1,c)]$
    \State \Return $c$
  \EndProcedure
  \State
  \Procedure{OEval}{$f,\mathcal{I}$}
    \State $\vec r \gets (T_i)_{i\in\mathcal{I}}$
    \State $m_0' \gets f((r.m_0)_{r\in\vec r})$
    \State $m_1' \gets f((r.m_1)_{r\in\vec r})$
    \State $c' \gets Eval(f,(r.c)_{r\in\vec r})$
    \State $T \gets T \doubleplus [(m_0',m_1',c')]$
    \State \Return $c'$
  \EndProcedure
  \State
  \Procedure{ODec}{$i$}
    \State \textbf{assert} $d<q$
    \State $d \gets d+1$
    \State \textbf{assert} $i<|T|$
    \If{$T_i.m_0\ne T_i.m_1$}
      \State \Return $\mathsf{None}$
    \Else
      \State $y\gets\mathsf{SampleCurrent}(T_i)$
      \State \Return $\mathsf{Some}(y)$
    \EndIf
  \EndProcedure
\end{algorithmic}
\end{algorithm}

The adversary may interleave encryption, evaluation, and decryption calls
adaptively.
Calls beyond the \(q\)-query decryption budget, out-of-range table indices,
and other failed assertions abort that execution.
In SSProve, an aborted execution contributes no probability to winning.

Here \(\mathsf{SampleCurrent}\) is instantiated by the surrounding game.
For a row whose two recorded plaintexts agree,
\(T_i=(m,m,\mathsf{Some}(\bar c,e))\), the real game uses
\(\mathsf{Flood}_e(\mathsf{dec}_0(sk,T_i.c))\), whereas the
simulated-decryption game uses \(\mathsf{Flood}_e(m)\), exactly as in
\Cref{eq:real-sim-decrypt}.
The formal proof factors the common prefix of \textsc{ODec} from this final
sampling continuation.

\subsection{Reduction adversary}

\begin{algorithm}
\caption{IND-CPA reduction adversary $\mathcal{B}_{\cA,q}$}
\label{alg:ind-cpa-reduction}
This is the conventional reduction endpoint used in the main proof.
It stores the same logical table as the IND-CPAD challenger, but obtains
challenge encryptions from the outer IND-CPA encryption oracle.
\begin{algorithmic}[1]
  \Procedure{$\mathsf{Flood}$}{$m,e$}
    \State $\sigma_e\gets\max\{1,e\}\gamma$
    \State $\eta \gets D_{\bbZ^n,0,\sigma_e^2}$
    \State \Return $J_m(\eta)$
  \EndProcedure
  \Statex
  \Procedure{$\mathcal{B}_{\cA,q}$}{$pk,evk$}
    \State $T \gets [\,]$
    \State $d \gets 0$
    \State $b^* \gets
      \cA^{OEnc_{\mathcal{B}},OEval_{\mathcal{B}},ODec_{\mathcal{B}}}(pk,evk)$
    \State \Return $b^*$
  \EndProcedure
  \Statex
  \Procedure{$OEnc_{\mathcal{B}}$}{$m_0,m_1$}
    \State $c \gets OEnc_{\mathsf{CPA}}(m_0,m_1)$
      \Comment{outer IND-CPA encryption oracle}
    \State $T \gets T \doubleplus [(m_0,m_1,c)]$
    \State \Return $c$
  \EndProcedure
  \Statex
  \Procedure{$OEval_{\mathcal{B}}$}{$f,\mathcal{I}$}
    \State $\vec r \gets (T_i)_{i\in\mathcal{I}}$
    \State $m_0' \gets f((r.m_0)_{r\in\vec r})$
    \State $m_1' \gets f((r.m_1)_{r\in\vec r})$
    \State $c' \gets Eval(f,(r.c)_{r\in\vec r})$
    \State $T \gets T \doubleplus [(m_0',m_1',c')]$
    \State \Return $c'$
  \EndProcedure
  \Statex
  \Procedure{$ODec_{\mathcal{B}}$}{$i$}
    \State \textbf{assert} $d<q$
    \State $d \gets d+1$
    \State \textbf{assert} $i<|T|$
    \State $(m_0,m_1,c) \gets T_i$
    \If{$m_0\ne m_1$}
      \State \Return $\mathsf{None}$
    \EndIf
    \State \textbf{assert} $c=\mathsf{Some}(\_,e)$
    \State $y \gets \mathsf{Flood}(m_0,e)$
    \State \Return $\mathsf{Some}(y)$
  \EndProcedure
\end{algorithmic}
\end{algorithm}

\FloatBarrier

The reduction does not know the IND-CPA challenge bit.
This is harmless for decryption queries that return a value: the simulator
only returns a flooded plaintext when the two recorded plaintexts agree,
$m_0=m_1$, so there is no bit-dependent plaintext choice left to make.

\section{SSProve Syntax and Semantics}
\label[appendix]{app:ssprove-substrate}

This appendix recalls the fragment of SSProve on which the program logic is
built~\cite{ssprove}.
SSProve embeds pure expressions shallowly in Rocq and
represents effects by a free monad.
Fix a result type \(A\).
Its raw code has
the following constructors:
\begin{equation}
\label{eq:ssprove-syntax}
\begin{split}
c ::= {}& \mathsf{return}(a)
 \mid \mathsf{call}(p,x,\kappa)
 \mid \mathsf{get}(\ell,\kappa)\\
 &{}\mid \mathsf{put}(\ell,v,c)
 \mid \mathsf{sample}(D,\kappa).
\end{split}
\end{equation}
Here \(a:A\).
An operation signature \(p:S\to T\) contains a procedure
identifier as well as its argument and result types; a call has \(x:S\) and
\(\kappa:T\to\mathsf{RawCode}(A)\).
A typed location
\(\ell:T_\ell\) is read into a continuation
\(\kappa:T_\ell\to\mathsf{RawCode}(A)\), or updated with \(v:T_\ell\).
Finally, sampling takes \(D\in\mathsf{SD}(X)\) and
\(\kappa:X\to\mathsf{RawCode}(A)\), where \(\mathsf{SD}(X)\) denotes discrete
subdistributions on \(X\).
Continuations are Rocq functions, so conditionals,
pure local computation, and structurally terminating loops are supplied by
the ambient language rather than by additional constructors.
Monadic bind
recursively pushes a continuation through this syntax.
We use the customary
notations
\begin{equation}
\label{eq:ssprove-notation}
\begin{aligned}
 x\leftarrow c_1;c_2
   &\equiv\mathsf{bind}(c_1,\lambda x.c_2),\\
 x\leftarrow p(a);c
   &\equiv\mathsf{call}(p,a,\lambda x.c),\\
 x\leftarrow\mathsf{get}\,\ell;c
   &\equiv\mathsf{get}(\ell,\lambda x.c),\\
 \mathsf{put}\,\ell:=v;c
   &\equiv\mathsf{put}(\ell,v,c),\\
 x\leftarrow D;c
   &\equiv\mathsf{sample}(D,\lambda x.c).
\end{aligned}
\end{equation}

An interface is a finite set of typed operation signatures.
The checked type
\(\mathsf{Code}_{L,I}(A)\) pairs raw code with a proof that every accessed
location belongs to the footprint \(L\) and every external call belongs to
the import interface \(I\).
A package \(P:I\to E\) is a finite map from the
procedure names in its export interface \(E\) to well-scoped implementations
\(S\to\mathsf{Code}_{L,I}(T)\).
Sequential composition \(P\circ Q\)
links \(P\)'s calls to procedures implemented by \(Q\).
At code level, the
essential clause is
\begin{equation}
\label{eq:ssprove-link}
\begin{aligned}
 &\mathsf{codeLink}(\mathsf{call}(p,a,\kappa),Q)\\
 &\qquad =
 x\leftarrow Q.p(a);
 \mathsf{codeLink}(\kappa(x),Q),
\end{aligned}
\end{equation}
with \(\mathsf{codeLink}\) acting recursively on the remaining constructors.
Parallel composition combines disjoint export interfaces, while the identity
package forwards every call.
Package validity tracks the import, export, and
memory-footprint side conditions used by these operations.

Each location carries a type and a default initial value.
A heap \(h\) maps
locations to values of their declared types; lookup of an unset location
returns its default.
After linking has eliminated external calls, fully linked code denotes a
state-transforming subdistribution
\begin{equation}
\label{eq:ssprove-denotation}
 \semantic{c}:\mathsf{Heap}\longrightarrow
  \mathsf{SD}(A\times\mathsf{Heap}).
\end{equation}
Writing \(\delta_z\) for the point distribution at \(z\), its defining
equations are
\begin{align}
 \semantic{\mathsf{return}(a)}_h
   &=\delta_{(a,h)},\notag\\
 \semantic{\mathsf{get}(\ell,\kappa)}_h
   &=\semantic{\kappa(h(\ell))}_h,\notag\\
 \semantic{\mathsf{put}(\ell,v,c)}_h
   &=\semantic{c}_{h[\ell\mapsto v]},\notag\\
 \semantic{\mathsf{sample}(D,\kappa)}_h
   &=\sum_{x}D(x)\,\semantic{\kappa(x)}_h.
\label{eq:ssprove-effect-semantics}
\end{align}
The sums are pointwise sums of subdistributions.
In particular, semantic
bind threads both the returned value and the updated heap:
\begin{align}
 \semantic{x\leftarrow c;\kappa(x)}_h
 &=
 \sum_{(a,h')}
   \semantic{c}_h(a,h')\,
   \semantic{\kappa(a)}_{h'}.
\label{eq:ssprove-bind-semantics}
\end{align}
SSProve defines failure by sampling the null subdistribution, and
\(\mathsf{assert}(b)\) as \(\mathsf{return}(())\) when \(b\) holds and failure
otherwise.
Thus failed assertions have the zero subdistribution and their
continuations are not run.
A fully linked game, with no unresolved oracle calls, is evaluated by resolving
its \(\mathsf{Run}\) procedure, starting from the default heap, and projecting
the result component of \Cref{eq:ssprove-denotation}.

\section{Checked Rules Used by the Security Proof}
\label{app:checked-rules}

\Cref{fig:security-rule-inventory} gives the complete program-logic interface
used directly by the proofs under \texttt{theories/Security}.
This interface consists of the exported lemmas in
\texttt{theories/ProgramLogics} whose names end in \texttt{Rule} and that the
security development invokes.
The figure remains schematic: it suppresses SSProve type indices and writes
\(\mu^c_{x,h}=\semantic{c(x)}_h\), \(\mathsf{wt}(\mu)\) for distribution
weight, and \(\mathsf{Eq}\) for equality of completed output--heap pairs
(including failure).
We use \(E_\#D\) for the pushforward of \(D\) by \(E\), and
\(\supp(\mu)\subseteq Q\) to mean that \(Q\) holds throughout the support.
The assertion \(M^{=}\) requires equal values and heaps satisfying \(M\), as
in \Cref{fig:logic-rules}; \(Q^{=}\) similarly requires \(Q\) on the left
output--heap pair and equality of both pairs.
All displayed side conditions are pointwise over inputs satisfying the
precondition.

\smallskip\noindent\emph{The raw additive-error judgment.}
This is a convenience wrapper around the completed judgment, not a second
coupling semantics.
For a postcondition \(Q\) on ordinary output--heap pairs, define
\[
 \begin{aligned}
 \widehat Q(\mathsf{Some}(z_L),\mathsf{Some}(z_R))
   &=Q(z_L,z_R),\\
 \widehat Q(\widehat z_L,\widehat z_R)
   &=\mathsf{false}\quad\text{otherwise}.
 \end{aligned}
\]
Then the judgment
\[
 \vDash\AEJ_{\rm raw}\{P\}\ c_L\approx_\epsilon c_R\ \{Q\}
\]
is defined to mean
\[
 \vDash\AEJ\{P\}\ c_L\approx_\epsilon c_R\ \{\widehat Q\}.
\]
Thus any failure outcome is bad, and the good event supplies concrete
intermediate values and heaps to a continuation.
The security proof uses this judgment only for the lossless, zero-error
challenge-initialization prefix: \textsc{AE-Raw-TVD} establishes its invariant,
\textsc{AE-Raw-Conseq} reshapes that invariant, and \textsc{AE-Seq} immediately
composes the result with the completed additive-error judgment for the rest of
the game.

This is an interface-level inventory, rather than a list of every helper
lemma in the compiler proof.
Several displayed rules are derived from the core rules in
\Cref{fig:logic-rules}: \textsc{DG-Vector} uses \textsc{KL-Sample}; the two
mixed rules use \textsc{Pyth-Seq} and \textsc{Pyth-Refl}; and
\textsc{Compile} uses \textsc{Pyth-Seq} throughout the call induction before
one final \textsc{Micciancio--Walter} conversion.

\begin{figure*}[!t]
\centering
\scriptsize

\textit{Unary Hoare and additive-error rules}
\begin{mathpar}
\inferrule*[right=\textsc{Hoare-Ret}]
  {P(x,h)\Longrightarrow Q(f(x),h)}
  {\vDash\mathsf{Hoare}\{P\}\ \mathsf{return}(f(x))\ \{Q\}}

\inferrule*[right=\textsc{AE-TVD}]
  {0\leq\epsilon \\
   \TVD(\completed{\mu^{c_L}_{x_L,h_L}},
        \completed{\mu^{c_R}_{x_R,h_R}})\leq\epsilon}
  {\vDash\AEJ\{P\}\ c_L\approx_\epsilon c_R\ \{\mathsf{Eq}\}}

\inferrule*[right=\textsc{AE-Raw-TVD}]
  {0\leq\epsilon \\
   \mathsf{wt}(\mu^{c_L}_{x_L,h_L})=
   \mathsf{wt}(\mu^{c_R}_{x_R,h_R})=1 \\
   \TVD(\mu^{c_L}_{x_L,h_L},\mu^{c_R}_{x_R,h_R})\leq\epsilon \\
   \supp(\mu^{c_L}_{x_L,h_L})\subseteq Q}
  {\vDash\AEJ_{\rm raw}\{P\}\ c_L\approx_\epsilon c_R\ \{Q^{=}\}}

\inferrule*[right=\textsc{AE-Conseq}]
  {P'\Longrightarrow P \\ Q\Longrightarrow Q' \\ \epsilon\leq\epsilon' \\
   \vDash\AEJ\{P\}\ c_L\approx_\epsilon c_R\ \{Q\}}
  {\vDash\AEJ\{P'\}\ c_L\approx_{\epsilon'}c_R\ \{Q'\}}

\inferrule*[right=\textsc{AE-Raw-Conseq}]
  {P'\Longrightarrow P \\ Q\Longrightarrow Q' \\ \epsilon\leq\epsilon' \\
   \vDash\AEJ_{\rm raw}\{P\}\ c_L\approx_\epsilon c_R\ \{Q\}}
  {\vDash\AEJ_{\rm raw}\{P'\}\ c_L\approx_{\epsilon'}c_R\ \{Q'\}}

\inferrule*[right=\textsc{AE-Triangle}]
  {P\Longrightarrow x_L=x_R\wedge h_L=h_R \\
   \vDash\AEJ\{P\}\ c_L\approx_\epsilon c_M\ \{\mathsf{Eq}\} \\
   \vDash\AEJ\{P\}\ c_M\approx_{\epsilon'}c_R\ \{\mathsf{Eq}\}}
  {\vDash\AEJ\{P\}\ c_L\approx_{\epsilon+\epsilon'}c_R\ \{\mathsf{Eq}\}}

\inferrule*[right=\textsc{AE-Seq}]
  {\vDash\AEJ_{\rm raw}\{P\}\ c_L\approx_\epsilon c_R\ \{M\} \\
   \vDash\AEJ\{M\}\ k_L\approx_{\epsilon'}k_R\ \{Q\}}
  {\vDash\AEJ\{P\}\
    (y\leftarrow c_L;k_L(y))
    \approx_{\epsilon+\epsilon'}
    (y\leftarrow c_R;k_R(y))\ \{Q\}}
\end{mathpar}

\textit{Pythagorean and derived sampling rules}
\begin{mathpar}
\inferrule*[right=\textsc{DG-Vector}]
  {E\ \mathsf{injective} \\ 0\leq\epsilon \\ 0<\sigma \\
   \displaystyle\frac{(v_i-u_i)^2}{2\sigma^2}\leq\epsilon
      \quad\text{for every }i \\
   P\Longrightarrow h_L=h_R \\
   P\Longrightarrow Q(y,h_L)\wedge Q(y,h_R)\quad\text{for every }y}
  {\vDash\Pyth\{P\}\
    z\leftarrow E_\#\mathsf{DG}^{n}(u,\sigma^2)
    \approx_{[n\epsilon]}
    z\leftarrow E_\#\mathsf{DG}^{n}(v,\sigma^2)\ \{Q\}}

\inferrule*[right=\textsc{Pyth-Refl}]
  {P\Longrightarrow x_L=x_R\wedge h_L=h_R \\
   0\leq s_i\quad\text{for every }i \\
   \supp(\mu^c_{x_L,h_L})\subseteq Q}
  {\vDash\Pyth\{P\}\ c\approx_s c\ \{Q\}}

\inferrule*[right=\textsc{Shared-Continuation}]
  {\vDash\Pyth\{P\}\ c_L\approx_s c_R\ \{M\} \\
   \vDash\mathsf{Hoare}\{M\}\ k\ \{Q\}}
  {\vDash\Pyth\{P\}\
    (y\leftarrow c_L;k(y))
    \approx_{s\doubleplus[0]}
    (y\leftarrow c_R;k(y))\ \{Q\}}

\inferrule*[right=\textsc{Shared-Prefix}]
  {P\Longrightarrow x_L=x_R\wedge h_L=h_R\wedge P_0(x_L,h_L) \\
   \vDash\mathsf{Hoare}\{P_0\}\ c\ \{M\} \\
   \vDash\Pyth\{M^{=}\}\ k_L\approx_s k_R\ \{Q\}}
  {\vDash\Pyth\{P\}\
    (y\leftarrow c;k_L(y))
    \approx_{[0]\doubleplus s}
    (y\leftarrow c;k_R(y))\ \{Q\}}
\end{mathpar}

\textit{Adaptive-call compiler rule}
\begin{mathpar}
\inferrule*[right=\textsc{Compile}]
  {\mathsf{ValidSeparatePreserving}(A,p,P,P',I) \\
   \vDash\Pyth\{I^{=}\}\ P.p\approx_s P'.p\ \{I\}}
  {\vDash\AEJ\{I^{=}\}\
   \mathsf{Compile}_q(A,p;P,P)
   \approx_{\sqrt{q\norm{s}_1/2}}
   \mathsf{Compile}_q(A,p;P',P)\ \{\mathsf{Eq}\}}
\end{mathpar}

\caption{Complete inventory of checked program-logic rules invoked directly
by the security development.
The labels correspond respectively to the Rocq theorems
\texttt{hoareRetRule}, \texttt{additiveErrorSameOutputTvdEqRule},
\texttt{additiveErrorRawTvdEqPostTotalRule},
\texttt{additiveErrorConseqRule}, \texttt{additiveErrorRawConseqRule},
\texttt{additiveErrorSameOutputTriangleRule},
\texttt{additiveErrorSeqRule}, \texttt{klDgNRule},
\texttt{pythReflRule}, \texttt{pythAeSeqRule},
\texttt{pythHoareSeqRule}, and \texttt{compileRule}.
Typing, support, and footprint premises are summarized as in
\Cref{fig:logic-rules}; \(\mathsf{ValidSeparatePreserving}\) groups the
compiler's validity, memory-separation, invariant-dependence, and
invariant-preservation premises.}
\label{fig:security-rule-inventory}
\end{figure*}
\FloatBarrier

\section{Additional Probability Details}
\label{sec:verified-kl-analysis}

This appendix records the distribution-level results behind the program rules
used in \Cref{sec:pyth-toolkit}.
It complements that section's program-level account: here we state the
probability facts available to the logic, while \Cref{sec:pyth-toolkit} explains
how SSProve programs consume them.

\subsection{Kullback-Leibler Infrastructure}

For MathComp distributions, the formalization defines KL divergence, absolute
continuity, and a predicate \(\mathsf{finiteKL}(\cP,\cQ)\) combining absolute
continuity with summability of the KL integrand.
The summability component matters in Rocq because many of the distributions in the development are countable rather than finite.
The core library also proves the log-sum inequality, KL preservation under injective pushforwards, and the summability facts used by the chain-rule proof.

\subsection{Pinsker and Pythagorean Preservation}

The formalization proves Pinsker's inequality for MathComp distributions.
\begin{lemma}[Pinsker's inequality]
    For all full distributions $\cP$ and $\cQ$ satisfying $\mathsf{finiteKL}(\cP,\cQ)$, we have
    \begin{equation}
        \Delta(\cP,\cQ)
  \leq
  \sqrt{\frac{\KL{\cP}{\cQ}}{2}}.
    \end{equation}
\end{lemma}
The proof reduces total variation to a binary event, proves the binary quadratic Pinsker inequality by real analysis, and combines it with the data-processing/log-sum argument for that event.

We also prove the chain-rule for the conditional-coordinate KL bound.
This is the result used by Pythagorean probability preservation.
\begin{lemma}
Let $\cP$ and $\cQ$ be distributions over $n$-tuples.
Assume
\begin{itemize}
    \item The distributions are lossless: $dweight(\cP)=dweight(\cQ)=1$.
    \item The relevant KL divergences are all finite:
    $\forall i,a.\ 
    \mathsf{finiteKL}(\cP_i\mid a,\cQ_i\mid a)$
    \item For all coordinate $i$ and history $a$, we have the following bound:
    $$\KL{\cP_i\mid a}{\cQ_i\mid a}\leq \varepsilon_i$$
\end{itemize}
We then have
\[
  \KL{\cP}{\cQ}\leq \sum_i \varepsilon_i .
\]
\end{lemma}

The proof is the technical heart of the probability development: it decomposes the KL integrand pointwise over tuple prefixes, controls both positive and negative parts of the resulting series, and then reassembles the bound.
The value-level inequality and the summability obligation are proved in parallel.
Both use the same grouping step: for each coordinate, traces are grouped by their prefix and current coordinate value, while the suffix is summed away.
For the value theorem this produces an expectation of conditional KL divergences; for summability the same grouping is applied to finite positive-part sums, with one negative-part slack term per coordinate.

Combining this chain bound with Pinsker gives the main preservation theorem.
\begin{lemma}[Pythagorean probability preservation]
\label{lem:verified-pythagorean-preservation}
Let $\cP,\cQ\in\distr(A^n)$ be distributions of weight one, and let
$s=(s_1,\ldots,s_n)$ be a tuple of nonnegative real numbers.
If, for every coordinate $i$ and every prefix $a\in A^{i-1}$,
\begin{align*}
  &\mathsf{finiteKL}(\cP_i\mid a,\cQ_i\mid a),\\
  &\KL{\cP_i\mid a}{\cQ_i\mid a}\leq s_i,
\end{align*}
then
\[
  \Delta(\cP,\cQ)
  \leq
  \sqrt{\frac{\sum_i s_i}{2}}.
\]
\end{lemma}
The same part of the formalization packages this lemma into a reusable Pythagorean distribution predicate, proves reflexivity and singleton-trace lemmas, and defines the call-error tuples used by the trace-compiler theorem.

\subsection{Discrete Gaussian Facts}

The discrete-Gaussian development is now one ingredient of the KL story rather than the whole story.
It proves the equal-variance KL formula used in \Cref{eq:dg-kl}:
\begin{lemma}[KL Divergence between discrete Gaussians]
\label{lem:dg-kl}
Let $\sigma\in\bbR$ be positive.
For all $\mu, \nu\in\bbZ$, we have
\[
  \KL{\DG{\mu}{\sigma^2}}{\DG{\nu}{\sigma^2}}
  =
  \frac{(\nu-\mu)^2}{2\sigma^2}.
\]
\end{lemma}
The proof is completely internal to the development.
It first exposes the normalized mass function, proves integer-shift invariance of the Gaussian weight, derives symmetry of the centered distribution, and then uses the zero-centered first moment to evaluate the remaining expectation.
The moment summability needed for that expectation is proved as well; it is no longer a side caveat.

The development also includes geometric upper bounds for events of the form
\[
  \Pr_{x\gets\DG{\mu}{\sigma^2}}
    \bigl[\, |x-\mu|\geq k\,\bigr].
\]
These tail bounds are useful for future cryptographic applications, but the current Pythagorean program logic primarily consumes the \Cref{lem:dg-kl} through the discrete-Gaussian sampling rule.

Thus the probability layer supplies three kinds of facts to the program logic:
ordinary distribution-level KL bounds such as the discrete-Gaussian theorem,
global KL-to-TVD conversion via Pinsker and the chain rule,
and completion/coupling lemmas that turn TVD bounds into additive-error
judgments over SSProve semantics.
The additive-error probability layer defines completion of a subdistribution
by adding an explicit failure point, constructs overlap and residual
distributions, and proves a maximal-coupling theorem from a TVD bound:
\begin{align*}
  \Delta(\mu,\nu)\leq \varepsilon
  \quad\Longrightarrow\quad&
  \exists d.\ \mathsf{Cpl}(d;\mu,\nu)\\
  &{}\wedge
  \Pr_{(x,y)\gets d}[x=y]\geq 1-\varepsilon .
\end{align*}
The program-level additive-error rules use this theorem after packing SSProve
outputs and heaps into a common completed output space.

\section{Formal Interface Excerpts}
\label{app:formal-interface-excerpts}

This appendix records selected Rocq excerpts that determine the meaning of the
main theorem.
They are not intended to replace the full artifact; rather, they expose the
small specification surface that a cryptographic reader should be able to
audit independently.

\begingroup
\setcoqmonofont

\subsection{Approximate Homomorphic Encryption Scheme Interface}
\label{app:approximate-homomorphic-encryption-interface}

\begin{minted}{coq}
Module Type ApproxFheScheme.
  Parameter pk_t : choice_type.
  Parameter evk_t : choice_type.
  Parameter sk_t : choice_type.
  Parameter message : choice_type.
  Parameter encryption : choice_type.
  (* Here we consider "tagged ciphertexts".
   * That is, an encryption together with an error bound.
   *
   * The `None` ciphertext should *only* come from evaluating unsupported operations.
   * e.g. out of circuit depth. *)
  Definition ciphertext := 'option (encryption × 'nat).
  (* We assume the homomorphic encryption operates over arithmetic circuits.
   * We therefore have a set of gates for building such circuits. *)
  Parameter unary_gate : choice_type.
  Parameter binary_gate : choice_type.
  Parameter interpret_unary : unary_gate → message → message.
  Parameter interpret_binary : binary_gate → message → message → message.
  (* Now, the "usual" 4-tuple (keygen, enc, eval, dec). *)
  Parameter keygen : distr R (pk_t × evk_t × sk_t).
  Axiom keygen_lossless : dweight keygen = 1.
  Parameter encrypt : pk_t → message → distr R ciphertext.
  Parameter eval1 : evk_t → unary_gate → ciphertext → distr R ciphertext.
  Parameter eval2 : evk_t → binary_gate → ciphertext → ciphertext →
    distr R ciphertext.
  Parameter decrypt : sk_t → ciphertext → distr R message.
\end{minted}

\subsection{Metric Charts}
\label{app:metric-chart-excerpt}

\begin{minted}{coq}
Module Type ApproxFheMetric(Import Scheme: ApproxFheScheme).
  Parameter metric : message → message → nat.
  Parameter dim : nat.
  (* We only care about metrics that are locally isometric to Z^n.
   * e.g., polynomials of some fixed degree whose coefficients belong to a finite field. *)
  (* Charts are origin-centered: each chosen center maps to the zero vector. *)
  Parameter isometry : message -> message -> dim.-tuple int.
  Parameter inverse_isometry : message -> dim.-tuple int -> message.
  Axiom isometry_center0 :
    forall (center : message), isometry center center = ivec_zero.
  Axiom metric_chartE :
    forall (center m : message),
    metric center m = ivec_dist ivec_zero (isometry center m).
  Axiom inverse_isometry_shift :
    forall (centerL centerR : message) (v : dim.-tuple int),
    inverse_isometry centerR v =
    inverse_isometry centerL (ivec_add v (isometry centerL centerR)).
\end{minted}

\subsection{Probability-One (Support-Level) Correctness}
\label{app:support-correctness-excerpt}

\begin{minted}{coq}
Module Type ApproxCorrectnessPerfect
  (Import Scheme: ApproxFheScheme) (Import M: ApproxFheMetric(Scheme)).
  (* For simplicity, we consider only pure and deterministic decryption. *)
  Parameter (deterministic_dec : sk_t → ciphertext → message).
  Axiom deterministic_dec_correct :
    ∀ sk c, \P_[ decrypt sk c ] (fun dec_out => ((dec_out == (deterministic_dec sk c)) : bool)) = 1.
  Definition is_underlying_plaintext sk (c : ciphertext) m :=
    match c with
    | None => false
    | Some (data, error_bound) => Order.le (metric (deterministic_dec sk c) m) error_bound
    end.
  Parameter (good_keys : pk_t → evk_t → sk_t → bool).
  Axiom keygen_perfect_correct :
    let bad_keys (keys : pk_t × evk_t × sk_t) :=
      let '(pk, evk, sk) := keys in ~~ (good_keys pk evk sk)
    in
    \P_[ keygen ] bad_keys = 0.
  Axiom encrypt_perfect_correct :
    ∀ pk evk sk m,
    good_keys pk evk sk →
    let bad_encryption c :=
        ~~ (is_underlying_plaintext sk c m)
    in
    \P_[ encrypt pk m ] bad_encryption = 0.
  Axiom eval1_perfect_correct :
    ∀ pk evk sk op c m,
    good_keys pk evk sk →
    is_underlying_plaintext sk c m →
    let bad_eval eval_out :=
      ~~ (is_underlying_plaintext sk eval_out (interpret_unary op m)) in
    \P_[ eval1 evk op c ] bad_eval = 0.
  Axiom eval2_perfect_correct :
    ∀ pk evk sk op c1 c2 m1 m2,
    good_keys pk evk sk →
    is_underlying_plaintext sk c1 m1 →
    is_underlying_plaintext sk c2 m2 →
    let bad_eval eval_out :=
      ~~ (is_underlying_plaintext sk eval_out (interpret_binary op m1 m2))
    in
    \P_[ eval2 evk op c1 c2 ] bad_eval = 0.
\end{minted}

\subsection{Discrete-Gaussian Distribution and KL Formula}
\label{app:discrete-gaussian-definition-excerpt}

The following excerpts pin down the sampler's mass function.  In particular,
the distribution is not an abstract parameter: it is the normalized
exponential weight below (and is the null subdistribution when the standard
deviation is nonpositive).

\begin{minted}{coq}
(* Unnormalized Gaussian function *)
Definition gaussian (s : R) (x : int) : R :=
  expR (- (x%:~R / s) ^ 2 / 2).
\end{minted}

\begin{minted}{coq}
(* Works "as expected" if s > 0.
 * null distribution otherwise. *)
Definition gaussian_pdf (s : R) (x : int) : R :=
  if s > 0 then
    gaussian s x / sum (gaussian s)
  else 0.
\end{minted}

\begin{minted}{coq}
Lemma isdistr_gaussian (s : R) :
  isdistr (gaussian_pdf s).
\end{minted}

\begin{minted}{coq}
Definition centered_discrete_gaussian s : distr R int :=
  mkdistr (isdistr_gaussian s).

Definition discrete_gaussian center s : distr R int :=
  dmargin (GRing.add center) (centered_discrete_gaussian s).
\end{minted}

The equal-variance KL identity consumed by the program logic is itself a
proved theorem, rather than an additional Gaussian assumption.

\begin{minted}{coq}
Theorem kl_discrete_gaussian (mu nu : int) (s : R) :
  s > 0 ->
  δ_KL (discrete_gaussian mu s) (discrete_gaussian nu s) =
    ((nu - mu)%:~R) ^+ 2 / (2 * s ^ 2).
\end{minted}

\subsection{Product Samplers and SSProve Bridge}
\label{app:discrete-gaussian-product-excerpt}

The construction folds the one-coordinate distribution into an independent
tuple.

\begin{minted}{coq}
Fixpoint dtuple {n : nat} {t: choiceType} :
    n.-tuple (distr R t) -> distr R (n.-tuple t) :=
  match n with
  | 0 => fun _ => dunit [tuple]
  | S i => fun ds =>
    \dlet_(x <- thead ds)
    \dlet_(xs <- (dtuple (behead_tuple ds)))
      dunit (cons_tuple x xs)
  end.

Definition nfold_distr (n : nat) {t: choiceType} (d: distr R t) :
  distr R (n.-tuple t) :=
  dtuple (nseq_tuple n d).
\end{minted}

\begin{minted}{coq}
Definition n_dg (n : nat) (s : R) : distr R (n.-tuple int) :=
  nfold_distr n (centered_discrete_gaussian s).
\end{minted}

The reduction uses SSProve's integer and vector types.  Its adapter transports
the same one-dimensional distribution into that integer type, then samples
the vector recursively.

\begin{minted}{coq}
Definition ssp_dg (m : 'int) (s : R) : distr R 'int :=
  dmargin (U := 'int) Z_of_int (discrete_gaussian (int_of_Z m) s).
\end{minted}

\begin{samepage}
\begin{minted}{coq}
Fixpoint discrete_gaussians_aux {n : nat} (s : R)
  : chVec 'int n -> distr R (chVec 'int n) :=
  match n with
  | 0 => fun _ => dunit (T := chVec 'int 0) tt
  | S i => fun ms =>
    let '(mhead, mtail) := ms in
    let dg := discrete_gaussians_aux s mtail in
    \dlet_(x <- ssp_dg mhead s)
    \dlet_(xs <- dg)
    dunit (x, xs)
  end.

Definition discrete_gaussians {n : nat} (center : chVec 'int n) (s : R) :=
  discrete_gaussians_aux s center.
\end{minted}
\end{samepage}

The proved bridge identifies that SSProve sampler, after conversion, with the
tuple sampler used by the transformed scheme.

\begin{minted}{coq}
  Lemma dmargin_toIntVec_discrete_gaussians_zero (n : nat) (s : R) :
    dmargin (@toIntVec n) (discrete_gaussians (zeroChVec n) s) =1
      n_dg n s.
\end{minted}

Finally, the addition in the displayed decryption code below is coordinatewise.

\begin{minted}{coq}
Definition ivec_add {n} (v w : n.-tuple int) : n.-tuple int :=
  [tuple (tnth v i + tnth w i) | i < n].
\end{minted}

\subsection{IND-CPA Game}
\label{app:ind-cpa-game-excerpt}

\begin{minted}{coq}
  Definition IndCpaOracle : IndCpaOracle_t :=
    [package IndCpa_locs ;
      #def #[oracle_encrypt] (messages : 'message_pair) : 'ciphertext
      {
        let (m0, m1) := messages in
        b ← get bit_addr ;;
        let m := if b then m1 else m0 in
        o ← get pk_addr ;;
        #assert isSome o as opk ;;
        let pk := getSome o opk in
        c <$ (ciphertext; encrypt pk m) ;;
        ret c
      }
    ].

  Definition IndCpaAdv_t := package IndCpaAdv_import IndCpaAdv_export.

  Definition IndCpaChallenger_t := package
    [interface
      #val #[adv_guess] : 'adv_keys → 'bool
    ]
    [interface
      #val #[main] : 'unit → 'bool
    ].

  Definition IndCpaChallenger : IndCpaChallenger_t :=
    [package IndCpa_locs ;
      #def #[main] (_ : 'unit) : 'bool
      {
        b <$ ('bool; dflip (1 / 2)) ;;
        keys <$ (pk_t × evk_t × sk_t; keygen) ;;
        let '(pk, evk, sk) := keys in
        #put bit_addr := b ;;
        #put pk_addr := Some pk ;;
        #put evk_addr := Some evk ;;
        b' ← call [ adv_guess ] : { pk_t × evk_t ~> 'bool} (pk, evk) ;;
        ret (eq_op b' b)
      }
    ].

  Definition IndCpaGame (Adv : nom_package) : nom_package :=
    ((IndCpaChallenger ∘ Adv)%sep ∘ IndCpaOracle)%share.
  
  Definition game_out (Adv : nom_package) : distr R bool :=
    dfst (Pr_op (IndCpaGame Adv) (main, ('unit, 'bool)) tt empty_heap).

  Definition winning_probability (A : nom_package) :=
    game_out A true.
End IndCpa.

Module Type IsIndCpa(Import Scheme: ApproxFheScheme).
  Module IndCpaGame := IndCpa Scheme.
  Import IndCpaGame.
  Parameter security_bound : nom_package -> R.
  Axiom is_secure : forall (A : nom_package),
    Package IndCpaAdv_import IndCpaAdv_export A ->
    winning_probability A <= security_bound A.
\end{minted}

\subsection{IND-CPA with Decryption Oracle and Game}
\label{app:ind-cpa-d-game-excerpt}

\begin{minted}{coq}
  Definition IndCpadOracle (max_queries : nat) : IndCpaOracle_t :=
    [package oracle_mem_spec ;
      #def #[oracle_encrypt] ('(m0, m1) : 'message × 'message ) : 'ciphertext
      {
        b ← get bit_addr ;;
        let m := if b then m1 else m0 in
        o ← get pk_addr ;;
        #assert isSome o as opk ;;
        let pk := getSome o opk in
        c <$ (ciphertext; encrypt pk m) ;;
        table ← get table_addr ;;
        let updated_table := (table ++ [ :: (m0,m1, c)]) in
        #put table_addr := updated_table ;;
        ret c
      } ; 
      #def #[oracle_eval1] ('(gate, r) : 'unary_gate × 'nat) : 'ciphertext
      {
        table ← get table_addr ;;
        #assert (r < length table) as r_in_range ;;
        let '(m0, m1, c) := nth_valid table r r_in_range in
        o ← get evk_addr ;;
        #assert isSome o as oevk ;;
        let evk := getSome o oevk in
        let m0' := interpret_unary gate m0 in
        let m1' := interpret_unary gate m1 in
        c' <$ (ciphertext; eval1 evk gate c) ;;
        let updated_table := (table ++ [ :: (m0', m1', c')]) in
        #put table_addr := updated_table ;;
        ret c'
      } ;
      #def #[oracle_eval2] ('((gate, ri), rj) : ('binary_gate × 'nat) × 'nat) : 'ciphertext
      {
        table ← get table_addr ;;
        #assert (ri < length table) as ri_in_range ;;
        #assert (rj < length table) as rj_in_range ;;
        let '(m0i, m1i, ci) := nth_valid table ri ri_in_range in
        let '(m0j, m1j, cj) := nth_valid table rj rj_in_range in
        let m0' := interpret_binary gate m0i m0j in
        let m1' := interpret_binary gate m1i m1j in
        o ← get evk_addr ;;
        #assert isSome o as oevk ;;
        let evk := getSome o oevk in
        c' <$ (ciphertext; eval2 evk gate ci cj) ;;
        let updated_table := (table ++ [ :: (m0', m1', c')]) in
        #put table_addr := updated_table ;;
        ret c'
      } ;
      #def #[oracle_decrypt] (i: 'nat) : 'option 'message
      {
        decrypt_count ← get decrypt_count_addr ;;
        #assert (decrypt_count < max_queries) ;;
        #put decrypt_count_addr := decrypt_count.+1 ;;
        table ← get table_addr ;;
        #assert (i < length table) as i_in_range ;;
        let '(m0, m1, c) := nth_valid table i i_in_range in
        if m0 == m1 then
          o ← get sk_addr ;;
          #assert isSome o as osk ;;
          let sk := getSome o osk in
          m <$ (message; decrypt sk c) ;;
          ret (Some m)
        else
          ret None
      }
    ].
\end{minted}

\begin{minted}{coq}
  Definition IndCpadAdv_t := package IndCpadAdv_import IndCpadAdv_export.

  Definition IndCpadChallenger_t := package
    [interface
      [guess] : { pk_t × evk_t ~> 'bool }
    ]
    [interface
      [main] : { 'unit ~> 'bool }
    ].

  Definition IndCpadChallenger : IndCpadChallenger_t :=
    [package oracle_mem_spec ;
      #def #[main] (_ : 'unit) : 'bool
      {
        b <$ ('bool; dflip (1 / 2)) ;;
        keys <$ (pk_t × evk_t × sk_t; keygen) ;;
        let '(pk, evk, sk) := keys in
        #put bit_addr := b ;;
        #put pk_addr := Some pk ;;
        #put evk_addr := Some evk ;;
        #put sk_addr := Some sk ;;
        b' ← call [ guess ] : { pk_t × evk_t ~> 'bool} (pk, evk) ;;
        ret (eq_op b' b)
      }
    ].

  Definition IndCpadGame
    (max_queries : nat) (Adv : nom_package) : nom_package :=
    ((IndCpadChallenger ∘ Adv)%sep ∘ IndCpadOracle max_queries)%share.

  Definition game_out
    (max_queries : nat) (Adv : nom_package) : distr R bool :=
    dfst (Pr_op (IndCpadGame max_queries Adv)
      (main, ('unit, 'bool)) tt empty_heap).

  Local Open Scope ring_scope.

  Definition winning_probability
    (max_queries : nat) (A : nom_package) :=
    game_out max_queries A true.

End IndCpad.

Local Open Scope ring_scope.

Module Type IsIndCpad(Import Scheme: ApproxFheScheme).
  Module IndCpadGame := IndCpad Scheme.
  Import IndCpadGame.
  (* Security bound may depend on the adversary and the decrypt-query budget. *)
  Parameter security_bound : nom_package -> nat -> R.
  Axiom is_secure : forall (A : nom_package) max_queries,
    Package IndCpadAdv_import IndCpadAdv_export A ->
    winning_probability max_queries A <= security_bound A max_queries.
\end{minted}

\subsection{Noise-Flooded Construction}
\label{app:noise-flooded-construction-excerpt}

\begin{minted}{coq}
Definition noise_flooding_dg_stdev
    (gaussian_width_multiplier : R) (error_bound : nat) : R :=
  (maxn 1 error_bound)%:~R * gaussian_width_multiplier.

Module Type NoiseFloodingParams.
Parameter gaussian_width_multiplier : R.
Axiom gt0_gaussian_width_multiplier :
  gaussian_width_multiplier > 0.
End NoiseFloodingParams.

Module NoiseFlooding
  (Import Scheme : ApproxFheScheme)
  (Import Metric : ApproxFheMetric(Scheme))
  (Import Params : NoiseFloodingParams)
  <: ApproxFheScheme.
Definition pk_t := Scheme.pk_t.
Definition evk_t := Scheme.evk_t.
Definition sk_t := Scheme.sk_t.
Definition Scheme_t := Scheme.Scheme_t.
Definition message := Scheme.message.
Definition encryption := Scheme.encryption.
Definition ciphertext := Scheme.ciphertext.
Definition unary_gate := Scheme.unary_gate.
Definition binary_gate := Scheme.binary_gate.
Definition interpret_unary := Scheme.interpret_unary.
Definition interpret_binary := Scheme.interpret_binary.
Definition keygen := Scheme.keygen.
Lemma keygen_lossless : dweight keygen = 1.
Proof. exact: Scheme.keygen_lossless. Qed.
Definition encrypt := Scheme.encrypt.
Definition eval1 := Scheme.eval1.
Definition eval2 := Scheme.eval2.
Definition dg_stdev (error_bound : nat) : R :=
  noise_flooding_dg_stdev gaussian_width_multiplier error_bound.
(* Invalid ciphertexts are modeled by the null subdistribution; the scheme
   interface treats them as unsupported evaluation results. *)
Definition decrypt (sk: sk_t) (c: ciphertext) : distr R message :=
  match c with
  | None => dnull
  | Some (_, e) =>
    \dlet_(m <- Scheme.decrypt sk c)
    \dlet_(noise <- n_dg dim (dg_stdev e))
    dunit (inverse_isometry m (ivec_add noise (isometry m m)))
  end.
\end{minted}

\subsection{Concrete IND-CPA Reduction}
\label{app:ind-cpa-reduction-excerpt}

\begin{minted}{coq}
  Definition IndCpadOracle (max_queries: nat) : IndCpaSim_t :=
    [package oracle_mem_spec ;
      #def #[oracle_encrypt] (messages : 'message_pair) : 'ciphertext
      {
        ready ← get ready_addr ;;
        #assert ready ;;
        c ← call [ oracle_encrypt ] : { message_pair ~> ciphertext } messages ;;
        table ← get table_addr ;;
        let '(m0, m1) := messages in
        let updated_table := (table ++ [ :: (m0, m1, c)]) in
        #put table_addr := updated_table ;;
        @ret ciphertext c
      } ; 
      #def #[oracle_eval1] (a : 'adv_ev1) : 'ciphertext
      {
        ready ← get ready_addr ;;
        #assert ready ;;
        let (gate, r) := a in
        table ← get table_addr ;;
        #assert (r < length table) as r_in_range ;;
        let '(m0, m1, c) := nth_valid table r r_in_range in
        o ← get evk_addr ;;
        #assert isSome o as oevk ;;
        let evk := getSome o oevk in
        let m0' := interpret_unary gate m0 in
        let m1' := interpret_unary gate m1 in
        c' <$ (ciphertext; eval1 evk gate c) ;;
        let updated_table := (table ++ [ :: (m0', m1', c')]) in
        #put table_addr := updated_table ;;
        @ret ciphertext c'
      } ;
      #def #[oracle_eval2] (a : 'adv_ev2) : 'ciphertext
      {
        ready ← get ready_addr ;;
        #assert ready ;;
        let '(gate, ri, rj) := a in
        table ← get table_addr ;;
        #assert (ri < length table) as ri_in_range ;;
        #assert (rj < length table) as rj_in_range ;;
        let '(m0i, m1i, ci) := nth_valid table ri ri_in_range in
        let '(m0j, m1j, cj) := nth_valid table rj rj_in_range in
        let m0' := interpret_binary gate m0i m0j in
        let m1' := interpret_binary gate m1i m1j in
        o ← get evk_addr ;;
        #assert isSome o as oevk ;;
        let evk := getSome o oevk in
        c' <$ (ciphertext; eval2 evk gate ci cj) ;;
        let updated_table := (table ++ [ :: (m0', m1', c')]) in
        #put table_addr := updated_table ;;
        @ret ciphertext c'
      } ;
      #def #[oracle_decrypt] (i: 'nat) : 'option_message
      {
        ready ← get ready_addr ;;
        #assert ready ;;
        decrypt_count ← get decrypt_count_addr ;;
        #assert (decrypt_count < max_queries) ;;
        #put decrypt_count_addr := decrypt_count.+1 ;;
        table ← get table_addr ;;
        #assert (i < length table) as i_in_range ;;
        let '(m0, m1, c) := nth_valid table i i_in_range in
        if (m0 == m1) then
          #assert isSome c as c_valid ;;
          let '(_, error_bound) := getSome c c_valid in
          noise <$ (chVec chInt dim;
            discrete_gaussians (zeroChVec dim)
              (noise_flooding_dg_stdev gaussian_width_multiplier error_bound)) ;;
          let res := inverse_isometry m0 (ivec_add (toIntVec noise) (isometry m0 m0)) in
          @ret ('option message) (Some res)
        else
          @ret ('option message) (None)
      }
    ].


  Definition IndCpaSimTop_t := package
    IndCpadAdv_export
    IndCpadAdv_export.

  Definition IndCpaSimTop : IndCpaSimTop_t :=
    [package oracle_mem_spec ;
      #def #[guess] ('(pk, evk) : 'adv_keys) : 'bool
      {
        ready ← get ready_addr ;;
        #assert (~~ ready) ;;
        #put ready_addr := true ;;
        #put pk_addr := Some pk ;;
        #put evk_addr := Some evk ;;
        b ← call [ guess ] : { pk_t × evk_t ~> 'bool } (pk, evk) ;;
        ret b
      }
    ].

  Definition IndCpaReduction (A : nom_package) (max_queries: nat) : nom_package :=
    ((IndCpaSimTop ∘ A)%sep ∘ IndCpadOracle max_queries)%share.
\end{minted}

\subsection{Loss Formula}
\label{app:loss-formula-excerpt}

\begin{minted}{coq}
Definition security_loss
    (dim max_queries : nat) (gaussian_width_multiplier : R) : R :=
  Num.sqrt
    ((max_queries%:R *
      (dim%:R / (2 * gaussian_width_multiplier ^+ 2))) / 2).
\end{minted}

\subsection{Reduction Validity}
\label{app:reduction-bound-excerpt}

\begin{minted}{coq}
  (* The package-level reduction preserves the IND-CPA adversary interface. *)
  Lemma ind_cpa_reduction_valid (A : nom_package) max_queries :
    Package IndCpaDSim.IndCpadAdv_import IndCpaDSim.IndCpadAdv_export A ->
    Package IndCpaSecurity.IndCpaGame.IndCpaAdv_import
      IndCpaSecurity.IndCpaGame.IndCpaAdv_export
      (ind_cpa_reduction A max_queries).
\end{minted}

\endgroup

\end{document}